\documentclass[11pt,a4paper]{article}

\usepackage[margin=3cm]{geometry}
\usepackage{hyperref}
\usepackage[svgnames]{xcolor}
\hypersetup{colorlinks=true,urlcolor=blue,linkcolor=DarkBlue,citecolor=DarkGreen}

\usepackage{cmap}

\usepackage[T1]{fontenc} 

\usepackage{graphicx} 

\usepackage{subcaption}

\usepackage{algorithm}
\usepackage[noend]{algorithmic}

\newcommand{\poly}{\textup{poly}}
\newcommand{\eps}{\ensuremath{\varepsilon}\xspace}
\newcommand{\ch}{\textsc{Consensus-Halving}\xspace}
\newcommand{\ich}{\ensuremath{I_{\textsc{CH}}}\xspace}
\newcommand{\dfms}{\ensuremath{I_{\textsc{CH}}^{\textsc{DFMS}}}\xspace}

\newcommand{\scut}{\textsc{Square}\xspace}
\newcommand{\isc}{\ensuremath{I_{\textsc{SC}}}\xspace}

\newcommand{\pizza}{\textsc{Pizza-Sharing}\xspace}

\def\bu/{\textup{\textsf{BU}}}
\def\linBU/{\textup{\textsf{LinearBU}}}
\def\bbu/{\textup{\textsf{BBU}}}
\def\fixp/{\textup{\textsf{FIXP}}}
\def\tfnp/{\textup{\textsf{TFNP}}}
\def\linearfixp/{\textup{\textsf{Linear-FIXP}}}
\def\ppad/{\textup{\textsf{PPAD}}}
\def\ppa/{\textup{\textsf{PPA}}}
\def\etr/{\textup{\textsf{$\exists\mathbb{R}$}}}
\def\fetr/{\textup{\textsf{FETR}}}
\def\tfetr/{\textup{\textsf{TFETR}}}
\def\betr/{\textup{\ensuremath{\textsf{ETR}_{[0,1]}}}}
\def\np/{\textup{\textsf{NP}}}
\def\fnp/{\textup{\textsf{FNP}}}
\def\pspace/{\textup{\textsf{PSPACE}}}
\def\p/{\textup{\textsf{P}}}

\newcommand{\bfeas}{\ensuremath{\textsc{Feasible}_{[0,1]}}\xspace}

\newcommand{\tucker}{\textsc{Tucker}\xspace}
\newcommand{\calI}{\ensuremath{\mathcal{I}}\xspace}

\newcommand{\calL}{\ensuremath{\mathcal{L}}\xspace}

\newcommand{\reals}{\ensuremath{\mathbb{R}}\xspace}
\newcommand{\naturals}{\ensuremath{\mathbb{N}}\xspace}

\newcommand{\area}{\ensuremath{\text{area}}}
\newcommand{\lplus}{``$+$''\xspace}
\newcommand{\lminus}{``$-$''\xspace}
\newcommand{\rplus}{\ensuremath{R^+}\xspace}
\newcommand{\rminus}{\ensuremath{R^-}\xspace}

\newcommand{\spizza}{\textsc{Straight-Pizza-Sharing}\xspace}
\newcommand{\dspizza}{\textsc{Discrete-Straight-Pizza-Sharing}\xspace}

\newcommand{\dpizza}{\textsc{Discrete-Square-Pizza-Sharing}\xspace}

\providecommand{\ceil}[1]{\ensuremath{\left \lceil #1 \right \rceil }}
\providecommand{\floor}[1]{\ensuremath{\left \lfloor #1 \right \rfloor }}

\DeclareSymbolFont{yhlargesymbols}{OMX}{yhex}{m}{n}

\DeclareMathAccent{\trngl}{\mathord}{yhlargesymbols}{"E6}

\newcommand{\epsborsuk}{\ensuremath{\eps\textsc{-Borsuk-Ulam}}\xspace}

\newcommand{\myparagraph}[1]{\vskip0.5\baselineskip \noindent {\bf #1}}

\usepackage{booktabs} 
\usepackage{setspace}

\usepackage{xspace}

\usepackage{mdframed}
\usepackage[utf8]{inputenc}

\usepackage{esvect}
\usepackage{amsfonts}
\usepackage{amsmath}
\usepackage{amsthm}
\usepackage{float}

\usepackage{textcomp}
\usepackage{gensymb}

\usepackage[capitalise,nameinlink,noabbrev]{cleveref}

\usepackage{sidecap}

\newtheorem{theorem}{Theorem}
\newtheorem{lemma}[theorem]{Lemma}

\newtheorem{definition}[theorem]{Definition}
\newtheorem{proposition}[theorem]{Proposition}
\newtheorem{remark}[theorem]{Remark}
\newtheorem{claim}[theorem]{Claim}

\usepackage{comment}
\usepackage{fullpage}

\floatname{algorithm}{Algorithm}

\allowdisplaybreaks

\title{Pizza Sharing is PPA-hard}

\author{Argyrios Deligkas\thanks{Royal Holloway University of London, UK. 
		email: \texttt{argyrios.deligkas@rhul.ac.uk}} 
	\and John Fearnley\thanks{University of Liverpool, UK. 
		email: \texttt{john.fearnley@liverpool.ac.uk}} 
	\and Themistoklis Melissourgos\thanks{University of Essex, UK. 
		email: \texttt{themistoklis.melissourgos@essex.ac.uk}}
}
\date{\vspace{-1.0cm}}

\begin{document}
	
	\maketitle
	
	\begin{abstract}
		We study the computational complexity of finding a solution for the straight-cut and square-cut pizza sharing problems. We show that computing an \eps-approximate solution is \ppa/-complete for both problems, while finding an exact solution for the square-cut problem is \fixp/-hard. Our \ppa/-hardness results apply for any $\eps < 1/5$, even when all mass distributions consist of non-overlapping axis-aligned rectangles or when they are point sets, and our \fixp/-hardness result applies even when all mass distributions are unions of squares and right-angled triangles. We also prove that the decision variants of both approximate problems are \np/-complete, while the decision variant for the exact version of square-cut pizza sharing is \etr/-complete.
	\end{abstract}
	
	\newpage
	
	\tableofcontents
	
	\newpage
	
	\section{Introduction}
	
	\emph{Mass partition problems} ask to fairly divide measurable objects that
	are embedded into Euclidean space \cite{RS20}. Perhaps the most popular mass
	partition problem is the \emph{ham sandwich} problem, in which three masses
	are given in three-dimensional Euclidean space, and the goal is to find a single
	plane that cuts all three masses in half. Recently, there has been interest in 
	\emph{pizza sharing} problems, which are mass partition problems in the 	
	two-dimensional plane, and in this work we study the computational complexity
	of such problems.
	
	In the \emph{straight-cut} pizza sharing problem, we are given $2n$
	two-dimensional masses in the plane, and we are asked to find straight
	lines (see \cref{fig:cuts-and-paths-a} for a depiction) that simultaneously bisect all of the masses. It has been shown that this problem always has a solution when we have $n$ straight lines available: the first result on the topic showed that solutions always exist when $n = 2$ \cite{BPS19}, and this was subsequently extended to show existence for all $n$ \cite{HK20}.

	\begin{figure}
		\begin{subfigure}[b]{0.3\textwidth}
			\centering
			\includegraphics[width=0.7\linewidth]{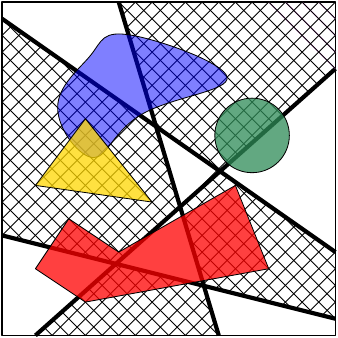}
			\caption{A set of straight-cuts with four lines.} \label{fig:cuts-and-paths-a}
		\end{subfigure}%
		\hspace*{\fill}   
		\begin{subfigure}[b]{0.3\textwidth}
			\centering
			\includegraphics[width=0.7\linewidth]{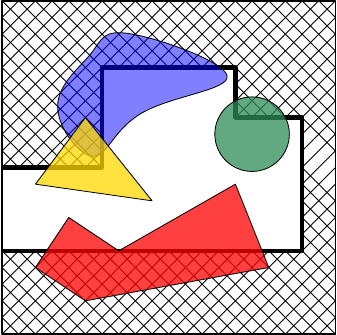}
			\caption{A square-cut-path with six turns (not $y$-monotone).} \label{fig:cuts-and-paths-b}
		\end{subfigure}%
		\hspace*{\fill}   
		\begin{subfigure}[b]{0.3\textwidth}
			\centering
			\includegraphics[width=0.7\linewidth]{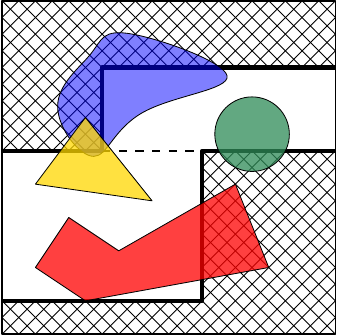}
			\caption{A $y$-monotone square-cut-path with four turns.} \label{fig:cuts-and-paths-c}
		\end{subfigure}%
		\caption{An example with 4 masses and various partitions of the plane into two regions, namely the shaded and non-shaded one. In a solution, each region contains half the area of each mass.}
		\label{fig:cuts-and-paths}
	\end{figure}	
	
	Another related problem is the \emph{square-cut} pizza sharing. In this problem,
	there are $n$ masses in the plane, and the task is to simultaneously
	bisect all masses using cuts, but the method of generating the cuts is
	different. Specifically, we seek a \emph{square-cut}, which consists of a single
	path that is the union of horizontal and vertical line segments. See
	\cref{fig:cuts-and-paths-b} and \cref{fig:cuts-and-paths-c} for two
	examples of square-cuts. Intuitively, we can imagine that a pizza cutter is
	placed on the plane, and is then moved horizontally and vertically without being
	lifted in order to produce the cut. Note that the path is allowed to wrap around
	in the horizontal axis: if it exits the left or right boundary, then it
	re-appears on the opposite boundary. So the cut in
	\cref{fig:cuts-and-paths-c} is still considered to be a single square-cut.
	
	It has been shown by \cite{KRS16} that, given $n$ masses, there always exists a \emph{square-cut-path} (termed \emph{\scut-path}) which makes at most $n-1$ turns and simultaneously bisects all of the
	masses. This holds even if the \scut-path is required to be
	\emph{$y$-monotone}, meaning that someone moving on the path would either never head South or never head North (e.g., \cref{fig:cuts-and-paths-c}).
	
	Two-dimensional fair division is usually called \textit{land division} in the literature. Land division is a prominent topic of interest in the Economics and AI communities that studies ways of fairly allocating two-dimensional objects among $n$ agents \cite{C05, SNHA17, SNHA20, ESS21, AD15, IH09, H11}. Popular mathematical descriptions of fair division problems first appeared in \cite{S48}, and since then, the existence of allocations under various fairness criteria have been extensively studied, together with algorithms that achieve them. These problems find applications from division of resources on land itself, to the Law of the Sea \cite{SS03-Consensus}, to redistricting \cite{L09, LS14}.
	
	Consensus halving is a problem that asks us to split a one-dimensional resource into two parts such that $n$ agents have equal value in both parts. Here, we study the same fairness criterion for $n$ agents, but for a two-dimensional resource. One can see that when we have the same fairness criterion at hand for any $k$-dimensional resource, $k \geq 2$, we can always translate the problem into its one-dimensional version, by integrating each agent's measure to a single dimension. Then a solution can be given by applying consensus halving. However, the solutions we get by doing so, are not taking into account the dimensionality of the problem, and as a result they might produce very unnatural solutions to a high-dimensional problem. For example, in land division, applying consensus halving would produce two parts, each of which can possibly be a union of $\ceil{n/2}$ disjoint land strips. \emph{Can we get better solutions by exploiting all the dimensions of the problem?}
	
	In this work we investigate different cutting methods of the two-dimensional objects, and in particular, two pizza sharing methods for which a solution is guaranteed. While based on intuition one might assume that exploiting the two dimensions would allow the complexity of finding a solution to be lower, our results show that this is not the case. We present polynomial time reductions from the one-dimensional problem to the two-dimensional problems showing that the latter ones are at least as hard as the former, i.e., \ppa/-hard.

	\myparagraph{Computational complexity of fair division problems.}
	There has been much interest recently in the computational complexity of fair
	division problems. In particular, the complexity class \ppa/ has risen to
	prominence, because it appears to naturally capture the complexity
	of solving these problems. For example, it has recently been shown by \cite{FRG18-Consensus, FRG18-Necklace} that the consensus halving problem, the ham
	sandwich problem, and the well-known necklace splitting problem are all
	\ppa/-complete. 
	
	More generally, \ppa/ captures all problems whose solution is verifiable in polynomial time and is guaranteed by the Borsuk-Ulam theorem. Finding an approximate solution to a Borsuk-Ulam function, or finding an exact solution to
	a linear Borsuk-Ulam function are both known to be \ppa/-complete
	problems \cite{Papadimitriou94-TFNP-subclasses, DFMS21}. The
	existence of solutions to the ham sandwich problem, the necklace splitting
	problem, and indeed the square-cut pizza sharing problem can all be proved via
	the Borsuk-Ulam theorem\footnote{It has also been shown by \cite{KA17} that the Borsuk-Ulam theorem is equivalent to the \emph{ham sandwich theorem} which states that the volumes of any $n$ compact sets in $\reals^n$ can always be simultaneously bisected by an $(n-1)$-dimensional hyperplane.}.
	
	\begin{theorem}[Borsuk-Ulam]
		\label{thm:borsuk-ulam}
		Let $f : S^d \to \reals^d$ be a continuous function, where $S^d$ is a
		$d$-dimensional sphere. Then, there exists an $x \in S^d$ 
		such that $f(x)=f(-x)$.
	\end{theorem}

    The other class of relevance here is the class \fixp/, defined by Etessami and Yannakakis \cite{EY10}. This is the class of problems that can
	be reduced in polynomial time to the problem of finding an exact fixed point of a Brouwer
	function. It is known by the aforementioned work, that $\fixp/$ contains the problem Square Root Sum, which has as input positive integers $a_1, \dots, a_n$ and $k$, and asks whether $\sum_{i=1}^{n} \sqrt{a_i} \leq k$. The question of whether Square Root Sum is in \np/ has been open for more than 40 years (\cite{GGJ76, P77, T92}). Furthermore,
	since there exist Brouwer functions that only have irrational fixed points, it
	is not expected that \fixp/ will be contained in \fnp/. In \cite{DFMS21}, it was shown that exact consensus halving is \fixp/-hard. 
	
	\myparagraph{Our contribution.} 
	We study the computational complexity of the straight-cut and square-cut pizza
	sharing problems,
	and we specifically study the cases where (i) all mass distributions are unions of weighted polygons (continuous version), and (ii) we are given unweighted point sets (discrete version). 
	We show that it is \ppa/-complete to find approximate solutions for the continuous and discrete versions of the two problems, while their decision variants are \np/-complete. Also, for the continuous version of the square-cut pizza sharing problem, we show that finding an exact solution is \fixp/-hard, while its decision variant is \etr/-complete.
	
	To the best of our knowledge, currently, there are no problems in computational geometry with \ppa/-hardness results other than discrete ham sandwich \cite{FRG18-Necklace}. We also note that pizza
	sharing problems do not need a circuit as part of the input, which makes them in
	some sense more ``natural'' than problems that are specified by circuits. Other
	known \ppa/-hard problems of this kind are one-dimensional, such as consensus
	halving~\cite{FRHSZ2020consensus-easier} and necklace splitting~\cite{FRG18-Necklace}, or problems with unbounded dimensions, such as discrete ham sandwich.
	Here we show the first \ppa/-hardness result for a ``natural'' two-dimensional
	problem. It is worth mentioning that shortly after the appearance of our result, Schnider in \cite{S21} proved the following: (a) the discrete version of straight-cut pizza sharing where each mass is represented by unweighted points is PPA-complete, and (b) for a more general input representation than ours, to find an exact solution in its continuous version is \fixp/-hard, and the decision variant is \etr/-complete.
	
	For both the straight-cut and the square-cut pizza sharing problems, namely \eps-\spizza and \eps-\scut-\pizza, we show that it is \ppa/-complete to find an
	$\eps$-approximate solution for any constant $\eps \in (0,1/5)$. This holds even when $n + n^{1-\delta}$ lines are permitted in a straight-cut pizza sharing instance with $2n$ mass distributions, and when $n - 1 + n^{1-\delta}$ turns of the square-cut path are permitted in a square-cut pizza sharing instance with $n$ mass distributions, for any constant $\delta \in (0,1]$. Furthermore, the \ppa/-hardness holds even when each mass distribution is uniform over polynomially many axis-aligned rectangles, and there is no overlap between any two mass distributions. The inapproximability for such high values of \eps is possible due to a recent advancement in the inapproximability of consensus halving \cite{DFHM25}. 
    
    The \ppa/ membership of straight-cut and square-cut pizza sharing holds even for inverse polynomial and inverse exponential \eps, respectively, and for weighted polygons with holes (arguably, a very general type of allowed input). To prove the \ppa/ membership of the square-cut pizza sharing problem, we first turn the original topological proof by \cite{KRS16} into an algorithmic one.
    Furthermore, we show that there is a constant $\eps > 0$ such that it is \np/-complete to decide whether an $\eps$-approximate solution of straight-cut pizza sharing with at most $n-1$ lines (resp. an $\eps$-approximate solution of square-cut pizza sharing with at most $n-2$ turns) exists. All of these results hold also for the discrete version of the problems. 

	We then turn our attention to the computational complexity of finding an
	\emph{exact} solution to the square-cut problem. We show that the problem
	of finding a \scut-path with at most $n-1$ turns that exactly bisects $n$
	masses is \fixp/-hard. This hardness result applies even if all
	mass distributions are unions of weighted axis-aligned squares and right-angled
	triangles. In order to prove this, we reduce from the problem of finding an exact \ch solution \cite{DFMS21}. Regarding the decision variant of the square-cut problem, we show that deciding whether there exists an exact solution with at most
	$n-2$ turns is \etr/-complete, where \etr/ consists of every decision problem that
	can be formulated in the existential theory of the reals (see \cref{sec:preliminaries} for its definition). All of our hardness results are summarized in \cref{tbl:straight-results} and \cref{tbl:results}.

	From a technical viewpoint, our \ppa/ membership result for straight-cut pizza sharing is based on a reduction that transforms mass distributions to point sets in general position and then employs a recent result by \cite{S21}. For the membership results of square-cut pizza sharing, our proof strategy is different, since we are able to directly reduce it to the \epsborsuk problem (see \cref{def: eps-Borsuk-Ulam_prob}). Our
	hardness results are obtained by reducing from the consensus halving
	problem, historically the first fair-division problem shown to be \ppa/-complete \cite{FRG18-Consensus}.

	\myparagraph{Further related work.}
	Since mass partitions lie in the intersection of Topology, Discrete Geometry, and Computer Science, there are several surveys on the topic; \cite{blagojevic2018topology,de2019discrete,Mat03BorsukUlam,ziv17} focus on the topological point of view, while \cite{agarwal1999geometric,edelsbrunner2012algorithms, kaneko2003discrete,matousek2013lectures} focus on computational aspects. Consensus halving \cite{SS03-Consensus} is the mass partition problem that received the majority of attention in the area of Economics and Computation so far \cite{DFRH,filos2018hardness,FRG18-Necklace,FRHSZ2020consensus-easier,filos2020topological}. Recently, Haviv \cite{H22} showed \ppa/-completeness of finding fair independent sets on cycle graphs, having as a starting point the latter problem.
	
	\begin{table}
		\centering
		\scalebox{1.0}{
                \begin{tabular}{lccccc}
                \hline
                \multicolumn{1}{c}{Hardness} & \eps & Lines            & Pieces    & Overlap & Theorem \\ \hline
                \multicolumn{6}{c}{Point sets}                                                               \\ \hline
                \ppa/                          & 1/5        & $n+n^{1-\delta}$ & -         & -       & \ref{thm: discr-straight-pizza}       \\
                \np/                           & $c$        & $n-1$            & -         & -       & \ref{thm: discr-straight-pizza-np-hard}       \\
                \multicolumn{6}{l}{}                                                                         \\ \hline
                \multicolumn{6}{c}{Mass distributions}                                                       \\ \hline
                \ppa/                          & 1/5        & $n+n^{1-\delta}$ & $\poly(n)$ & $1$     & \ref{thm:straight-pizza-PPA}       \\
                \np/                           & $c$        & $n-1$            & $\poly(n)$ & $1$     & \ref{thm:straight-np-h-overlap}      
                \end{tabular}
                }
		\caption{A summary of our hardness results for \eps-\spizza. Here, $c$ and $\delta$ are absolute, positive constants. 
			``Lines'' refers to the number of cut-lines allowed in a solution. 
			``Pieces'' refers to the maximum number of distinct polygons that define every mass distribution.
			``Overlap'' refers to the maximum number of different mass distributions allowed to contain any point of $[0,1]^2$.}
		\label{tbl:straight-results}
	\end{table}

        \begin{table}
		\centering
		\scalebox{1.0}{
            \begin{tabular}{lccccc}
            \hline
            \multicolumn{1}{c}{Hardness} & \eps & Turns              & Pieces    & Overlap & Theorem \\ \hline
            \multicolumn{6}{c}{Point sets}                                                                 \\ \hline
            \ppa/                          & 1/5        & $n-1+n^{1-\delta}$ & -         & -       & \ref{thm: discr-sc-pizza}       \\
            \np/                           & $c$        & $n-2$              & -         & -       & \ref{thm: discr-sc-pizza-np-hard}       \\
            \multicolumn{6}{l}{}                                                                           \\ \hline
            \multicolumn{6}{c}{Mass distributions}                                                         \\ \hline
            \ppa/                          & 1/5        & $n-1+n^{1-\delta}$ & $\poly(n)$ & $1$     & \ref{thm:hvu-pizza-ppa-h}       \\
            \np/                           & $c$        & $n-2$              & $\poly(n)$ & $1$     & \ref{thm:square-np-h-overlap}       \\ \hline
            \fixp/                         & 0          & $n-1$              & $6$       & $3$     & \ref{thm:fixp-h}       \\
            \etr/                          & 0          & $n-2$              & $6$       & $3$     & \ref{thm:etr-h}      
            \end{tabular}
		}
		\caption{A summary of our hardness results for \eps-\scut-\pizza. Here, ``turns'' refers to the number of turns a solution (\scut-path) is allowed to have. The definitions of $c$, $\delta$ and the semantics of ``pieces'', and ``overlap'' are the same as those of \cref{tbl:straight-results}.}
		\label{tbl:results}
	\end{table}

	\section{Preliminaries}\label{sec:preliminaries}
	
	\myparagraph{Mass distributions and point sets.}
	A \emph{mass distribution} $\mu$ on $[0,1]^2$ is a measure on the plane such that all open subsets of $[0,1]^2$ are measurable, $0 < \mu\left([0,1]^2\right) < \infty$, and $\mu(S) = 0$ for every subset of $[0,1]^2$ with dimension lower than 2. For any given $d \in \naturals^*$ we denote $[d] := \{ 1, 2, \dots, d \}$, and we denote by $\bigsqcup$ the union of disjoint sets. For every $S \in [0,1]^2$ we denote by $\area(S)$ the Lebesgue measure of $S$ on $\reals^2$, i.e., the area of $S$.
	Let a mass distribution $\mu$ be described by a finite set of non-overlapping regions $a_1, a_2, \ldots, a_d$, i.e., $\bigsqcup_{j = 1}^{d} a_j = [0,1]^2$, such that $\sum_{j=1}^d \mu(a_j) = \mu \left( [0,1]^2 \right)$. Then, $\mu$ is \emph{piece-wise uniform} if for every $j$ and every
	$S \subseteq a_j$ it holds that $\mu(S) = w_j \cdot \area(S)$
	for some \emph{weight} $w_j > 0$ independent of $S$. When additionally $w_j = w_k$ for all $j,k \in [d]$ then the mass distribution is called \emph{uniform}. The \emph{support} of mass distribution $i \in [n]$, denoted by $supp(i)$, is the area $A_i \subseteq [0,1]^2$ which has the property that for every $S \subseteq A_i$ with $\area(S) > 0$ we have $\mu_{i}(S) > 0$.
	Let $N := \{ I \subseteq [n] : \bigcap_{i \in I} supp(i) \neq \emptyset \} $. A set of mass distributions $\mu_{1}, \dots, \mu_{n}$, or \emph{colours}, has \emph{overlap} $k$ if $\max_{I \in N} | I | = k $. 
        Finally, a mass distribution is \emph{normalised} if $\mu([0,1]^2) = 1$. For ease of presentation, all our additive approximation results on the continuous versions of the problems assume that all mass distributions are normalised, which is without loss of generality.

    A point set $P = (p_1, p_2, \dots, p_d)$ on $[0,1]^2$ consists of $d \in \naturals^*$ many non-overlapping point masses. Throughout this work, the points that will be considered in the discrete versions of our problems have the same finite weight, so when we partition them (by partitioning $[0,1]^2$), it suffices to measure the cardinality of the points in each part.

	\myparagraph{Set of straight-cuts.}
	A set of \emph{straight-cuts}, or \emph{cut-lines}, or simply \emph{lines} defines subdivisions of the plane $R$. \cref{fig:cuts-and-paths-a} shows an example of a set of straight-cuts. Each line creates two half-spaces, and arbitrarily assigns number ``$0$'' to one and ``$1$'' to the other. A subdivision of $R$ is labeled ``$+$'' (and belongs to \rplus) if its parity is odd (according to the labels given to the half-spaces) and ``$-$'' (and belongs to \rminus) otherwise. Observe that by flipping the numbers of two half-spaces defined by a line, we flip all the subdivisions' labels. Thus, there are only two possible labelings of the subdivisions. 
	
	\myparagraph{Square-cut-path.}
	A \emph{square-cut-path}, denoted for brevity \scut-path, is a non-crossing
	directed path that is formed only by horizontal and vertical line segments
	and in addition it is allowed to ``wrap around'' in the horizontal dimension.  
	\cref{fig:cuts-and-paths-b} and \cref{fig:cuts-and-paths-c} show two examples of \scut-paths.
	A \emph{turn} of the path is where a horizontal segment meets with a vertical 
	segment. A \scut-path is \emph{$y$-monotone} if
	all of its horizontal segments are monotone with respect to the $y$ axis. Any \scut-path partitions the plane $R$ into two regions, namely, \rplus and \rminus, so that the following holds: for any two points of the plane, if the straight line that connects them intersects once the path, then the two points have opposite labels.\footnote{Notice that a path can pass multiple times from the same point.}

	\myparagraph{Pizza sharing.}
	A set of lines (resp. a \scut-path) \emph{\eps-bisects} a mass
	distribution $\mu$, if $|\mu(\rplus) - \mu(\rminus)| \leq \eps$. It
	\emph{simultaneously \eps-bisects} a set of mass distributions $\mu_1, \dots, \mu_n$ if $|\mu_i(\rplus) - \mu_i(\rminus)| \leq \eps$ for every $i \in [n]$.

	\begin{mdframed}[backgroundcolor=white!90!gray,
		leftmargin=\dimexpr\leftmargin-20pt\relax,
		innerleftmargin=4pt,
		innertopmargin=4pt,
		skipabove=5pt,skipbelow=5pt]
		\begin{definition}
			\label{def:eps-pizza} For any $n \geq 1$, the problem \emph{\eps-\spizza} is defined as follows:
			\begin{itemize}
				\item \textbf{Input:} $\eps \geq 0$, and mass distributions $\mu_1, \mu_2, \ldots, \mu_{2n}$ on $[0,1]^2$.
				\item \textbf{Output:} A partition of $[0,1]^2$ to \rplus and \rminus using at most $n$ lines such that for each $i \in [2n]$ it holds that $|\mu_i(\rplus) - \mu_i(\rminus)| \leq \eps$.
			\end{itemize}
		\end{definition}
	\end{mdframed} \vspace{5pt}
	\begin{mdframed}[backgroundcolor=white!90!gray,
		leftmargin=\dimexpr\leftmargin-20pt\relax,
		innerleftmargin=4pt,
		innertopmargin=4pt,
		skipabove=5pt,skipbelow=5pt]
		\begin{definition}
			\label{def:eps-hv-pizza} For any $n \geq 1$, the problem \emph{\eps-\scut-\pizza} is defined as follows:
			\begin{itemize}
				\item \textbf{Input:} $\eps \geq 0$, and mass distributions $\mu_1, \mu_2, \ldots, \mu_{n}$ on $[0,1]^2$.
				\item \textbf{Output:} A partition of $[0,1]^2$ to \rplus and \rminus using a \scut-path with at most $n-1$ turns such that for each $i \in [n]$ it holds that $\left| \mu_i(\rplus) - \mu_i(\rminus) \right| \leq \eps$.
			\end{itemize}
		\end{definition}
	\end{mdframed} \vspace{5pt}
	
	In \cite{HK20} and \cite{KRS16} it was proved that \eps-\spizza and \eps-\scut-\pizza, respectively, always admit a solution for arbitrary absolutely continuous masses with respect to the Lebesgue measure (i.e., area), and for any $\eps \geq 0$ (see Theorem 1 of the former, and Theorem 2.4 of the latter work). While the aforementioned results hold for such general measures, for the computational problems \eps-\spizza and \eps-\scut-\pizza we need a standardized way to describe the input, and therefore restrict to particular classes of measures. We consider the class of mass distributions that are defined by weighted simple polygons with holes. This class consists of mass distributions that can be succinctly represented in the input of a Turing machine, while at the same time provide great expressive power. 
	
	In particular, we will use the standard representation of 2-dimensional simple polygons in computational geometry problems, that is, a directed chain of points\footnote{From this point on, whenever we refer to polygons we will implicitly assume that they are simple polygons.}. Consider a polygon that is defined by $k$ points $p_{i} = (x_i,y_i)$, where $x_i, y_i \in [0,1] \cap \mathbb{Q}$, for $i \in [k]$, which form a directed chain $C = (p_1, \dots, p_k)$. This chain represents a closed boundary defined by the line segments $(p_i, p_{i+1})$ for $i \in [k-1]$ and a final one $(p_{k}, p_{1})$. Since we consider polygons with holes, we need a way to distinguish between the polygons that define a boundary whose interior has strictly positive weight and polygons that define the boundary of the holes (whose interior has zero weight). We will call the former \textit{solid} and the latter \textit{hollow} polygon. To distinguish between the two, we define a solid polygon to be represented by directed line segments with counterclockwise orientation, while a hollow polygon to be represented similarly but with clockwise orientation. Furthermore, each solid polygon $C_s$, its weight $w$ and its $r \geq 0$ holes $C_{h_1}, C_{h_{2}}, \dots, C_{h_{r}}$ in the interior, are grouped together in the input to indicate that all these directed chains of points represent a single polygon $(w, C_s, C_{h_1}, \dots, C_{h_r})$.
	
	Although it is not hard to construct instances of \eps-\spizza (resp. \eps-\scut-\pizza) where $n$ lines (resp. $n-1$ turns for any \scut-path) are necessary in order to constitute a solution, there might be cases where a solution can be achieved with fewer lines (resp. a \scut-path with fewer turns). Hence, we also study the decision variant of these problems, in which we ask whether we can find a solution with at most $k$ lines (resp. $k$ turns), where $k < n$ (resp. $k < n-1$). Note also that, due to the normalization assumption of the considered measures, $\eps \in [0,1]$.

	\myparagraph{Consensus halving.}
	The main hardness results of this work are proved by
	reductions from the consensus halving problem. 
	
	In the \eps-\ch\ problem, there is a set of $n$ agents with \emph{valuation density functions} $v_i : [0,1] \to \reals_{\geq 0}$, $i \in [n]$. For any given interval $[a,b]$, let us denote $v_i ([a,b]) := \int_{a}^{b} v_i (x) \, dx$. The goal is to find a partition of the $[0,1]$ interval into subintervals labelled either \lplus or \lminus, using at most $n$ cuts. This partition should satisfy that for every agent $i$, the total value for the union of subintervals $\calI^{+}$ labelled \lplus and the total value for the union of subintervals $\mathcal{I}^{-}$ labelled \lminus is the same up to \eps, i.e., $|v_i(\calI^{+}) - v_i(\calI^{-})| \leq \eps$. Furthermore, in order for \eps to be meaningful, we consider normalized valuation functions, that is, $v_i ([0,1]) = 1$ for all $i \in [n]$, which implies that $\eps \in [0,1]$. In our results, we will use the following types of valuation functions (see \cref{fig:ch-vals} for a depiction).
	\begin{itemize}
		\item $k$-\emph{block}: consists of at most $k$ non-overlapping (but possibly adjacent) intervals $[a^\ell_{1}, a^r_{1}], \ldots, [a^\ell_{k}, a^r_{k}]$ where interval $[a^\ell_{j}, a^r_{j}]$ has density $c_{j} > 0$, and 0 otherwise. So, $v([a^\ell_{j},x]) = (x-a^\ell_{j})\cdot c_{j}$ for every $x \in [a^\ell_{j},a^r_{j}]$.
		\item \emph{$k$-block uniform}: $k$-block, where the density of every interval is $c > 0$ (same for all blocks).
		\item \emph{$k$-block-triangle}: union of a $k$-block valuation function and an extra interval $[a^\ell_{1}, a^r_{1}]$, where interval $[a^\ell_{1}, a^r_{1}]$ has density $2 \cdot (x - a^\ell_{1}) \cdot c_1$ for some $c_1 > 0$, therefore $v([a^\ell_{1},x]) = (x - a^\ell_{1})^2 \cdot c_1$ for every $x \in [a^\ell_{1}, a^r_{1}]$. Also, $(a^\ell_{1}, a^r_{1}) \cap [a^\ell_{j}, a^r_{j}] = \emptyset$ for every $j \in [k]$. 
	\end{itemize}

	\begin{figure}[htbp]
		\begin{subfigure}{0.31\textwidth}
			\centering
			\includegraphics[width=\linewidth]{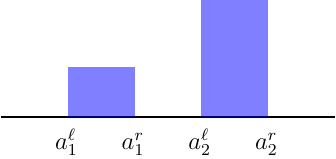}
			\caption{2-block valuation} \label{fig:ch-vals-a}
		\end{subfigure}%
		\hspace*{\fill}   
		\begin{subfigure}{0.31\textwidth}
			\centering
			\includegraphics[width=\linewidth]{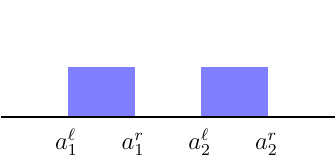}
			\caption{2-block uniform valuation} \label{fig:ch-vals-b}
		\end{subfigure}%
		\hspace*{\fill}   
		\begin{subfigure}{0.31\textwidth}
			\centering
			\includegraphics[width=\linewidth]{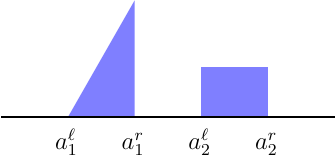}
			\caption{1-block-triangle valuation} \label{fig:ch-vals-c}
		\end{subfigure}%
		\caption{Some examples of valuation functions.}
		\label{fig:ch-vals}
	\end{figure}

	\myparagraph{Complexity classes.} 
	\eps-\spizza and \eps-\scut-\pizza are examples of \emph{total} problems, which are problems
	that always have a solution.
	The complexity class \tfnp/ (\textup{\textsf{Total Function}} \np/) defined in \cite{Megiddo1991},
	contains all total problems whose solutions can be verified in polynomial time. 
 
    In this work, we will focus on a well-known subclass of \tfnp/, namely \ppa/, defined by Papadimitriou \cite{Papadimitriou94-TFNP-subclasses}. This class captures problems whose totality is guaranteed by the \emph{parity argument} on undirected graphs: if there is an odd-degree vertex then there is another one. In the typical \ppa/ problem, \ensuremath{\textsc{EndOfUndirectedLine}}, we are given a Boolean circuit $N$ with input of size $n$ and output of size $2n$, and the circuit has a $\poly(n)$ size description. The input represents the identity of a vertex and the output represents the identities of (at most) two other vertices. If for two vertices $i, j$ we have $j \in N(i)$ and $i \in N(j)$, then we consider an undirected edge between them. This implies an undirected graph structure where the maximum degree of any vertex is 2. The problem is, given a vertex of degree 1, to find any other vertex of degree 1. Now notice that the graph size is $2^n$, whereas the input is $\poly(n)$ large, therefore, common algorithms that would solve the problem in case the graph was described explicitly are no longer useful. As discussed earlier, since the definition of \ppa/, many problems have been shown to be complete for the class, yet most of them require a circuit description in their input. The more interesting cases of \ppa/-completeness are for problems with more ``natural'' inputs, in the sense that they require no such circuit description. The pizza sharing problems we study here are among those ones.
	
	The complexity class \etr/ consists
	of all decision problems that can be formulated in the \emph{existential theory of the reals (ETR)} \cite{Matousek14, scha09}. In other words, problems that can be written in ETR form: $\exists \vec{P} \in \reals^{m} \cdot \Phi$, where $\Phi$ is a Boolean formula using connectives $\{\land, \lor, \lnot\}$ over polynomials with domain $\reals^{m}$ for some $m \in \mathbb{N}$ compared with the operators $\{ <, \leq, =, \geq, > \}$. It is known that $\np/ \subseteq \etr/ \subseteq
	\pspace/$ \cite{C88}, and it is generally believed that \etr/ is distinct from the
	other two classes. The class \fetr/ (\textup{\textsf{Function}} \etr/) consists of all search problems whose decision variant is in \etr/. The class $\tfetr/$ is the subclass of \fetr/ which contains only problems that admit a solution (i.e. all the instances of their decision variant are ``yes'' instances). Both \fetr/ and \tfetr/ were introduced in \cite{DFMS21} as the natural analogues of \fnp/ and \tfnp/ in the real RAM model of computation. For a definition of the real RAM model we refer the reader to the detailed work of Erickson, van der Hoog, and Miltzow \cite{EHM20}.

	In this work, our focus regarding complexity classes of \tfetr/ will be on the class \fixp/.
    \fixp/ was defined in \cite{EY10} and captures problems whose totality is guaranteed by Brouwer's fixed point theorem \cite{brouwer1911abbildung}. An instance of a typical problem in \fixp/ consists of the description of a continuous function $g : D \to D$, where $D$ is a nonempty, compact, and convex set. $g$ is represented by an \emph{algebraic circuit}, and a solution of the instance is any $x \in D$ such that $g(x) = x$. An algebraic circuit is a circuit that operates on real
	numbers, and uses gates from the set $\{c, +, -, \times c, \times, \max,
	\min\}$; a $c$-gate outputs the constant $c$, a $\times c$-gate multiplies
	the input by a constant $c$, and all other gates behave according to their
	standard definitions, where $c \in \mathbb{Q}$. It is worth noting that, since each of these gates' output is a continuous function of its input, any function $g$ constructed using those gates is continuous on $D$.

	\section{Hardness results}
	Here we show all hardness results regarding the exact and approximate versions of our pizza sharing problems for mass distributions, as well as for point sets. For the \ppa/- and \np/-hardness results on mass distributions, the instances we construct are such that there is no overlap between any two mass distributions. Notice that the case of non-overlapping mass distributions is the most simple type of an instance, since we can easily reduce it to one where an arbitrarily large number of masses overlap.\footnote{Given an \eps-\spizza (resp. \eps-\scut-\pizza) instance with $2n$ (resp. $n$) mass distributions, we can pick an arbitrary distribution $i$ and create an extra $(2n+1)$-st (resp. $(n+1)$-st) distribution by copying $i$. Then, a solution to the resulting instance is a solution to the initial instance, and vice versa. Notice that in the latter instance at least $2$ mass distributions overlap, and we can repeat this ``copying'' procedure as many times as needed to achieve any number of overlapping distributions.}

    Our \ppa/-hardness proofs of the continuous versions of \eps-\spizza and \eps-\scut-\pizza are via reductions from the \eps-\ch problem (\cref{thm:straight-pizza-PPA} and \cref{thm:hvu-pizza-ppa-h}, respectively). Consequently, by a general construction (\cref{lem: cont-to-discr}), we reduce those to their discrete versions to get \ppa/-hardness. We also show that the decision variants of the approximation problems are \np/-hard by using our \ppa/-hardness constructions to reduce from the respective decision variant of \ch from \cite{filos2018hardness} (\cref{thm:straight-np-h-overlap} and \cref{thm:square-np-h-overlap}, respectively), and the \np/-hardness is attained in the discrete problems too. Finally, for the exact version of \scut-\pizza, via reductions from exact \ch, we show that the problem is \fixp/-hard (\cref{thm:fixp-h}), while its decision variant is \etr/-hard (\cref{thm:etr-h}).
	
	\subsection{Hardness of approximate \spizza}
	\label{sec:straight-pizza}
	
	We start by proving that \eps-\spizza is \ppa/-hard for any $\eps < 1/5$, even for very simple mass distributions. We prove our result via a reduction from \eps-\ch with $k$-block valuations, which for the special case of $3$-block uniform valuations has been shown to be \ppa/-complete \cite{DFHM25}. In addition, we explain how to combine the machinery of our reduction with that of \cite{filos2018hardness} in order to get \np/-hardness for \eps-\spizza, where $\eps > 0$ is a small constant.

	\paragraph{\bf The reduction.} We reduce from \ch with $2n$ agents, and for each agent we create a corresponding mass in \spizza. Firstly, we finely discretize the $[0,1]$ interval into blocks and we place the blocks on $y=x^2$, where $x \geq 0$. So, the $[0,1]$ interval corresponds to a part of the quadratic equation. This guarantees that every line can cut this ``bent'' interval at most twice and in addition the part of each mass that is in $\rplus$ is almost the same as value of the corresponding agent for the piece of $[0,1]$ labelled with \lplus. 
 
    Next we show how to construct an instance $I_P$ of $(\eps - \eps')$-\spizza with $2n$ mass distributions, for any constant $r \geq 1$, and $1/n^r \leq \eps' < \eps \leq 1$, given an instance \ich of \eps-\ch with $2n$ agents with $k$-block valuations. 
	
	Let $c_{\max} := \max_{i \in [2n], m \in [k]} c_{im}$, where $c_{im}$ is the value density of agent $i$'s $m$-th block in \ich, and observe that $c_{\max} \geq 1$ since the total valuation of any agent over $[0,1]$ is $1$. In what follows, it will help us to think of the interval $[0,1]$ in \ich as being discretized in increments of $d := \frac{1}{\ceil{ 8 \cdot n \cdot c_{\max} / \eps'}}$. We refer to the subinterval $\left[ (j-1) \cdot d, j \cdot d \right]$ as the \emph{$j$-th $d$-block} of interval $[0,1]$ in \ich, for $j \in [1/d]$.

    We now describe the instance $I_P$. We consider two kinds of square tiles; $1/d$ large square tiles of size $\frac{d^2}{24} \times \frac{d^2}{24}$, each of which contains $2n$ smaller square tiles of size $\frac{1}{2n} \frac{d^2}{24} \times \frac{1}{2n} \frac{d^2}{24}$ on its diagonal. We will call the former type \textit{big-tile} and denote it by $t_j$ and the latter one \textit{small-tile} and denote it by $t_{ij}$ for some $i \in [2n]$, $j \in [1/d]$.

    For every agent $i \in [2n]$ of \ich we will create a uniform mass distribution $\mu_i$ that consists of at most $1/d$ many axis-aligned small-tiles. Each big-tile $t_j$ is centered at $(\frac{j d}{2}, \frac{j^2 d^2}{4})$, $j \in [1/d]$, and in it, each small-tile $t_{ij}$, $i \in [2n]$, belonging to mass distribution $\mu_i$ has its bottom left corner at $\left( \frac{j d}{2} - \frac{d^2}{48} + \frac{(i-1) d^2}{48 n}, \frac{j^2 d^2}{4} - \frac{d^2}{48} + \frac{(i-1) d^2}{48 n} \right)$. Each small-tile $t_{ij}$ contains total mass (belonging to $\mu_i$) of $v_{ij} \cdot \frac{2}{n} \left( \frac{d^2}{24} \right)^2$, where $v_{ij}$ is the total value that agent $i$ has for the $j$-th $d$-block in \ich. 
    Observe that, by definition, we have $v_{ij} \leq d \cdot c_{\max} \leq \frac{\eps'}{8 n} < \frac{1}{8 n}$, where the last inequality comes from the fact that $\eps' < 1$, and therefore $v_{ij} \cdot \frac{2}{n} \left( \frac{d^2}{24} \right)^2$ fits inside the small tile of size $\frac{1}{2n} \frac{d^2}{24} \times \frac{1}{2n} \frac{d^2}{24}$. In particular, the mass inside $t_{ij}$ has width $ \frac{1}{2n} \frac{d^2}{24}$ and height $v_{ij} \cdot \frac{4 d^2}{24}$. Finally, it is easy to check that all big-tiles are in $[0,1]^2$: by construction, the big-tiles are placed such that the $1$-st big-tile's bottom-left corner has the smallest $x$- and $y$-coordinates, while the $1/d$-th big-tile's top-right corner has the largest $x$- and $y$-coordinates among points that belong to mass distributions of $I_P$. The aforementioned points' coordinates are $\left( \frac{d}{2} - \frac{d^2}{48}, \frac{d^2}{4} - \frac{d^2}{48} \right) $, and $\left( \frac{1}{2} - \frac{d^2}{48} + \frac{d^2}{24}, \frac{1}{4} - \frac{d^2}{48} + \frac{d^2}{24} \right) = \left( \frac{1}{2} + \frac{d^2}{48}, \frac{1}{4} + \frac{d^2}{48} \right)$, respectively, and both are in $[0,1]^2$ since $d \leq 1$. \cref{fig:curve} and \cref{fig:construction} depict our construction.

        \begin{figure}
		\centering
		\includegraphics[scale=1.0]{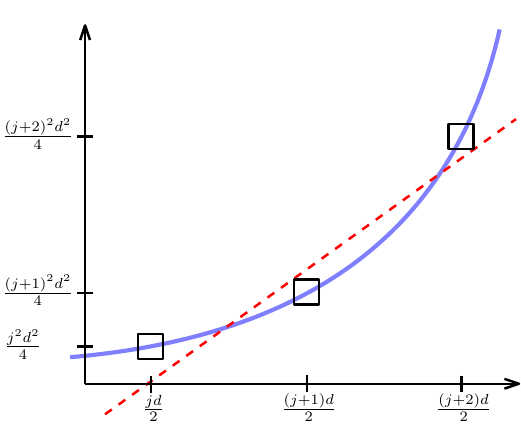}
		\caption{Placing the big-tiles on the $y=x^2$ curve. The $j$-th, $(j+1)$-st and $(j+2)$-nd big-tiles are centered on the curve. Their size is small enough to prevent any straight line (red/dashed) from intersecting more than two big-tiles.}
		\label{fig:curve}
	\end{figure}

	\begin{figure}[htbp]
		\begin{subfigure}[b]{0.35\textwidth}
			\centering
			\includegraphics[width=0.9\linewidth]{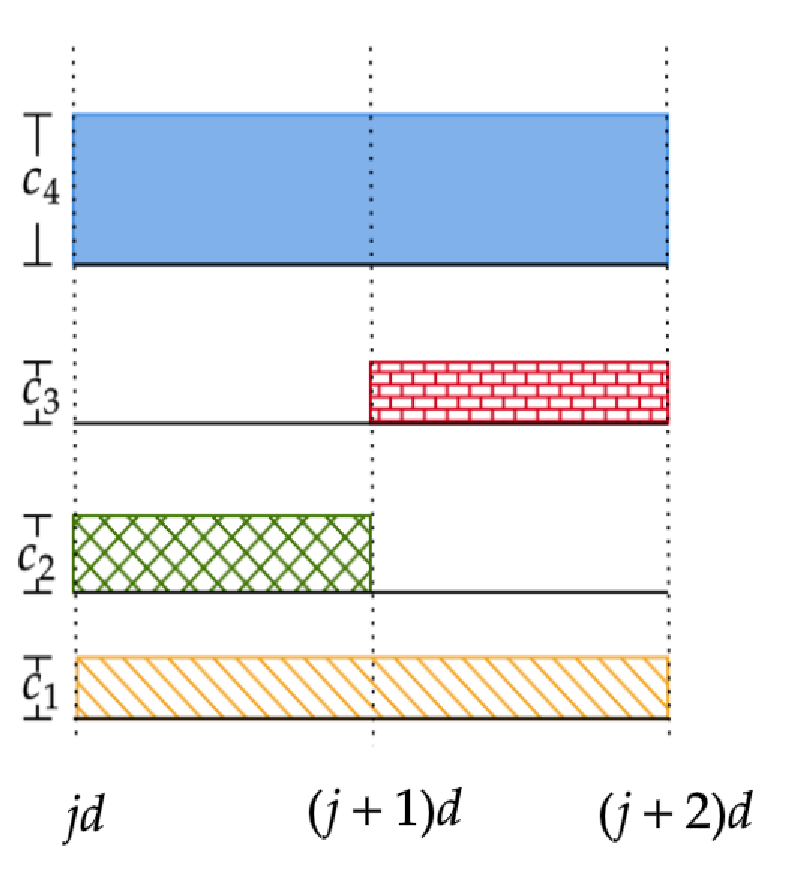}
			\caption{The $(j+1)$-st (left) and $(j+2)$-nd $d$-block (right) of the \ch instance. \\ \\ } \label{fig:construction-a}
		\end{subfigure}%
		\hspace*{\fill}   
		\begin{subfigure}[b]{0.6\textwidth}
			\centering
			\includegraphics[width=0.75\linewidth]{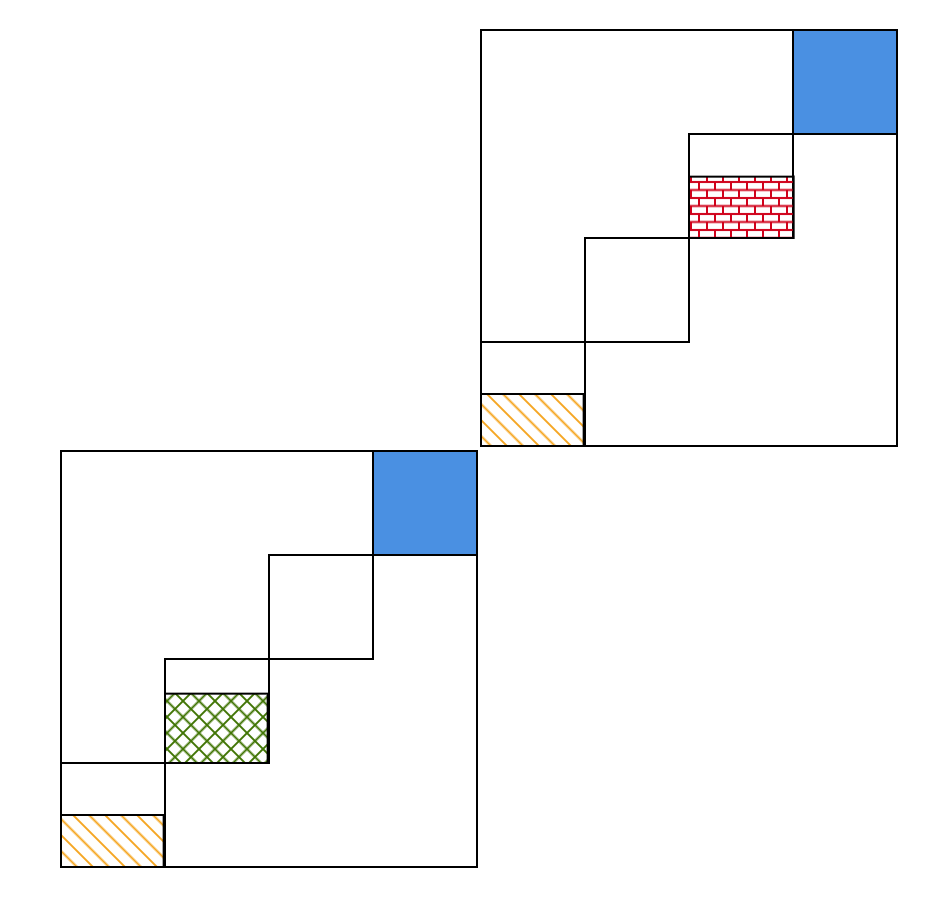}
			\caption{The $(j+1)$-st (left) and the $(j+2)$-nd (right) big-tile of the \spizza and the \scut-\pizza instance with their small-tiles. The small-tiles contain the mass distributions of the $(j+1)$-st and $(j+2)$-nd $d$-block, respectively.} \label{fig:construction-b}
		\end{subfigure}%
		\caption{The construction for a part of an instance with four mass distributions.}
		\label{fig:construction}
	\end{figure}

    Next, we prove the following auxiliary claim.

    \begin{claim}\label{clm: line-dist-three-points}
        Any straight line in $[0,1]^2$ cannot have distance at most $\frac{d^2}{24}$ with more than two centers of big-tiles.
    \end{claim}

    \begin{proof}
        For the sake of contradiction, suppose there are three tiles, $t_a , t_b , t_c$, with centers $p_j = (x_j , y_j)$, $j \in \{ a, b, c \}$, such that every $p_j$ has distance at most $\frac{d^2}{24}$ from a line $\ell$. Then, let us move $\ell$ in parallel until it passes through $p_a$. Next, rotate the line around $p_a$ until it passes through $p_b$ as well. These two movements of $\ell$ resulted in a new line $\ell'$ whose distance from $p_c$ is at most $3 \cdot \frac{d^2}{24} = \frac{d^2}{8}$, since each movement costed an extra distance of at most $\frac{d^2}{24}$.

        The distance between $\ell'$ and $p_c$ is
        \begin{align*}
		\frac{ \left| (y_{b} - y_{a})(x_{a} - x_{c}) - (y_{a} - y_{c})(x_{b} - x_{a}) \right| }{ \sqrt{(x_{b} - x_{a})^2 + (y_{b} - y_{a})^2}} = \frac{ \left| (d^2 / 4) (a-c) (b - c) \right| }{ \sqrt{1 + (d^2 / 4) (a + b)}},
	\end{align*}
        where the equality comes after substituting the coordinates of the centers of $t_a , t_b$ and $t_c$ as defined by our construction and simplifying the expression. Now recall that $a, b, c \in [1/d]$, which means that they are integers, and additionally, pairwise different, since otherwise they would be less than three distinct points. This means that we can bound from below the minimum distance between $\ell'$ and $p_c$ for $|a-c| = |b-c| = 1$ and $a = b = 1/d$. In other words, the minimum distance between $\ell'$ and $p_c$ is at least
        \begin{align*}
		\frac{ d^2 / 4 }{ \sqrt{1 + (d^2 / 4) (2/d)}} = \frac{ d^2 / 4 }{ \sqrt{1 + d / 2}} > \frac{d^2}{8} ,
	\end{align*}
        a contradiction.
    \end{proof}

        Now we are ready to prove the following.
        
	\begin{lemma}
		\label{lem:lines-to-cuts}
		Fix constants $\delta \in (0,1]$, $r \geq 1$, and let $1/n^r \leq \eps' < \eps \leq 1$.
		Let $\calL = \{\ell_{1}, \ldots,  \ell_{m}\}$ be a set of lines, where $m \leq n+n^{1-\delta}$. 
		If $\calL$ is a solution to $(\eps-\eps')$-\spizza instance $I_P$, then we can find in polynomial time a solution to \eps-\ch instance \ich with at most $2(n+n^{1-\delta})$ cuts.
	\end{lemma}

        \begin{proof}
            We will first prove that there is no line that intersects more than two big-tiles of $I_P$. This comes almost directly from \cref{clm: line-dist-three-points}. In particular, recall that each big-tile has size $\frac{d^2}{24} \times \frac{d^2}{24}$, therefore it fits inside a circle with radius $\frac{d^2}{24}$, whose center is the barycenter of the big-tile. If any line could intersect more than two big-tiles, then it would intersect also their corresponding circles, which means that it would have distance at most $\frac{d^2}{24}$, contradicting \cref{clm: line-dist-three-points}.

            Now, given a solution of $I_P$, we define the cuts and labels for the \ch instance, \ich. We consider the big-tiles of $I_P$ in sequential order and we add one cut at $j\cdot d$ whenever we find two big-tiles $t_j$ and $t_{j+2}$ that belong to different regions, i.e. \lplus, \lminus, or vice versa. This change of region can happen at most $2(n+n^{1-\delta})$ times. Hence, we have at most $2(n+n^{1-\delta})$ cuts in the instance \ich. Each $d$-block of \ich follows the label of its corresponding big-tile of the solution of $I_P$, except for those that correspond to intersected big-tiles. The latter $d$-blocks are arbitrarily given one of the two labels. 
            
            The aforementioned arbitrary labeling of the intersected big-tiles will cause some extra discrepancy in the solution of \ich. In particular, each such big-tile will be adding to each valuation
            $v_i$, $i \in [2n]$ of \ich, discrepancy $2 \cdot v_{ij} \leq 2 \cdot d \cdot c_{\max} \leq \frac{\eps'}{4 n}$, by construction of \ich. Since $|\calL| \leq n+n^{1-\delta}$, and each line of \calL can intersect two big-tiles, all lines of \calL collectively add discrepancy $v_i$ of value at most $2 (n+n^{1-\delta}) \cdot \frac{\eps'}{4 n} \leq \eps'$.

	Let us denote by $\calI^+$ and $\calI^-$ the regions in \ich as translated from the solution of $I_P$ according to the aforementioned process. Then, $|v_i(\calI^+) - v_i(\calI^-)| \leq |\mu_i(\rplus) - \mu_i(\rminus)| + \eps' \leq (\eps - \eps') + \eps' = \eps$. Therefore, this is a solution to \eps-\ch.
	
    Finally, notice that the construction of $I_P$ from \ich is a polynomial time reduction. In particular, since $\eps' \geq 1/n^r$ for some constant $r \geq 1$, we have that $1/d = \ceil{ 8 \cdot n \cdot c_{\max} / \eps' } \leq \ceil{ 8 \cdot n \cdot c_{\max} \cdot n^r }$ is a value polynomial in the input size, and the creation of each of the $1/d$ big-tiles can be done in polynomial time. 
        \end{proof}
	
    This, together with the fact that \eps-\ch is \ppa/-hard for any $\eps < 1/5$, due to \cite{DFHM25}, implies the main theorem of this section.

	\begin{theorem}\label{thm:straight-pizza-PPA}
		\eps-\spizza with $2n$ mass distributions is \ppa/-hard for any constant $\eps < 1/5$, even when $n+n^{1-\delta}$ lines are allowed for any given constant $\delta \in (0,1]$, every mass distribution is uniform over polynomially many rectangles, and there is no overlap between any two mass distributions.
	\end{theorem}

    We will now shift our attention to studying the decision variant of the problem, where we are asking to find a solution that uses at most $n-1$ straight lines, and notice that there is no guarantee for such a solution. We employ the \np/-hard instances of \eps-\ch for their constant \eps from \cite{filos2018hardness}, and reduce them according to the above reduction procedure to $(\eps - \eps')$-\spizza instances for some $\eps'$ that is inverse polynomial in the input size, e.g., $\eps' = 1/n^2$. This gives the following.

	\begin{theorem}\label{thm:straight-np-h-overlap}
		There exists a constant $\eps > 0$ for which it is \np/-hard to decide if an \eps-\spizza instance with $2n$ mass distributions admits a solution with at most $n-1$ lines, even when every mass distribution is uniform over polynomially many rectangles, and there is no overlap between any two mass distributions.
	\end{theorem}

	\subsection{Hardness of approximate \scut-\pizza} \label{sec: eps-sc-pizza_hardness}
	In this section, we prove hardness results for \eps-\scut-\pizza.
	We provide a reduction from \eps-\ch with $k$-block valuations, which was shown to be \ppa/-complete in \cite{DFHM25} for any constant $\eps < 1/5$ even for $3$-block uniform valuations. Also, the machinery that we present, combined with the reduction by \cite{filos2018hardness}, implies \np/-hardness for the decision variant of \eps-\scut-\pizza, where $\eps > 0$ is a small constant.
	
	\paragraph{\bf The reduction.} We reduce from a general \eps-\ch instance to an \eps-\scut-\pizza instance, and the idea is to create a mass for each agent. 
	For any constant $r \geq 1$, and $1/n^r \leq \eps' < \eps \leq 1$, given an instance \ich of \eps-\ch with $n$ agents and $k$-block valuations we will show a polynomial time construction to an $(\eps-\eps')$-\scut-\pizza instance \isc. 
	
	For our construction, we will use the same components as those in the proof of \cref{lem:lines-to-cuts}. In particular, let $c_{\max} := \max_{i \in [n], m \in [k]} c_{im}$, where $c_{im}$ is the value density of agent $i$'s $m$-th block in \ich (and again note that $c_{\max} \geq 1$ since the total valuation of the agent over $[0,1]$ is $1$). Similarly to the aforementioned proof, we will discretize the $[0,1]$ interval of \ich in increments of $d := \frac{1}{\ceil{4 \cdot n \cdot c_{\max} / \eps'}}$. Also, let us restate that for any given $j \in [1/d]$, the subinterval $\left[ (j-1) \cdot d, j \cdot d \right]$ is called $j$-th $d$-block of interval $[0,1]$ in \ich.
	
	We will be using the same gadgets that were constructed for the proof of \cref{lem:lines-to-cuts}, namely the big-tiles, which contain small-tiles. In particular, we have $1/d$ square big-tiles of size $d \times d$, each of which contains $n$ square small-tiles of size $\frac{d}{n} \times \frac{d}{n}$ on its diagonal. For any $i \in [n]$, $j \in [1/d]$, we denote them by $t_j$ and $t_{ij}$, respectively.
	
	In this construction, however, the positioning of big-tiles will be different than that of the aforementioned proof. In particular, we will be placing them on the diagonal of $[0, 1]^2$ as shown in \cref{fig:construction-b}.
	For every agent $i$ we will create a uniform mass distribution $\mu_i$ that consists of at most $1/d$ many axis-aligned small-tiles. For each $j \in [1/d]$, the bottom-left corner of big-tile $t_j$ is at $\left( (j-1)d, (j-1)d \right)$. In it, each small-tile $t_{ij}$, $i \in [n]$, belonging to mass distribution $\mu_i$ has its bottom left corner at $\left( (j-1)d + (i-1)\frac{d}{n}, (j-1)d + (i-1)\frac{d}{n} \right)$. 
	
	Inside a given big-tile $t_j$, each small-tile $t_{ij}$ contains total mass (belonging to $\mu_i$) of $v_{ij} \cdot \frac{4 d^2}{n}$, where $v_{ij}$ is agent $i$'s total value for the $j$-th $d$-block in \ich. 
	The mass inside the small-tile is rectangular, with width $\frac{d}{n}$ and height $v_{ij} \cdot 4 d$. Observe that $v_{ij} \leq d \cdot c_{\max} \leq \frac{\eps'}{4 n} < \frac{1}{4 n}$, where the first inequality is by definition of a $d$-block in \ich, the second one is by definition of $d$, and the last one is due to $\eps' < 1$. Therefore the total mass of $v_{ij} \cdot \frac{4 d^2}{n}$ fits inside the small-tile. By definition of the big-tiles positioning and size, it is straightforward that they are inside $[0,1]^2$.
	\cref{fig:construction} depicts our construction.
	
	Now we are ready to prove the following lemma.
	
	\begin{lemma}
		\label{lem:sc-paths-to-cuts}
		Fix constants $\delta \in (0,1]$, $r \geq 1$, and let $2/n^r \leq \eps' < \eps \leq 1$. Let a \scut-path with at most $n-1+n^{1-\delta}$ turns be a solution to $(\eps-\eps')$-\scut-\pizza instance \isc. Then we can find in polynomial time a solution to \eps-\ch instance \ich with at most $n+n^{1-\delta}$ cuts.
	\end{lemma}
	
	\begin{proof}
        We have to specify how a solution of \isc, i.e., a \scut-path, is translated back to a solution of \ich, i.e., a set of cuts. This is identical to the one used in the proof of \cref{lem:lines-to-cuts}. In particular, we consider again the big-tiles in sequential order and we add one cut at $j\cdot d$ whenever we find two big-tiles $t_j$ and $t_{j+2}$ that belong to different regions. Suppose that, following the aforementioned procedure, the next \ich cut falls at $j \cdot d'$ for some $d' > d$. If $t_j$ belongs to region \lplus (resp. \lminus) and $t_{j+2}$ belongs to \lminus (resp. \lplus), then the interval $[j \cdot d, j \cdot d']$ gets label \lplus (resp. \lminus), and vice versa. It is easy to see that this translation takes polynomial time.
        
        What remains is to prove that the translation of the aforementioned solution of a $(\eps-\eps')$-\scut-\pizza instance \isc into a solution of the \eps-\ch instance \ich is indeed correct. 
        Notice that, if the solution of \isc has $r \in \naturals$ many turns on the \scut-path, there can be at most $r+1$ small-tiles that are intersected by it (since there are $r+1$ line segments). For our reduction, let $r = n-1+n^{1-\delta}$ for any constant $\delta \in (0,1]$. Consider sequentially the big-tiles $t_1, \dots, t_{1/d}$, and without loss of generality, let $t'_1, \dots, t'_{r+1}$ be its subset, where $t'_j$ is the $j$-th big-tile that has an intersected small-tile. In the big-tile sequence of $t_1, \dots, t_{1/d}$, the change of region can happen at most $n+n^{1-\delta}$ times. Therefore, we have at most $n+n^{1-\delta}$ cuts in \ich with the corresponding labels as defined previously. Each $d$-block of \ich follows the label of the corresponding big-tile in \isc, except for those corresponding to intersected big-tiles. These $d$-blocks are given an arbitrary label.
        
		The above translation of the \scut-path to \ich cuts indicates that, each line segment of the \scut-path that intersects a big-tile, introduces a discrepancy between the \lplus and \lminus regions of \ich, of value at most $2 \cdot c_{\max}\cdot d \leq \frac{\eps'}{2n}$ for each valuation $v_i$, $i \in [n]$. Taking into account the entire \scut-path, this results in total discrepancy of at most $(n+n^{1-\delta}) \cdot \frac{\eps'}{2n} \leq \eps'$ for each agent $i$.
		
        We now denote by $\calI^+$ and $\calI^-$ the regions in \ich according to the above translation from a solution of \isc to a solution of \ich. We have
			$|v_i(\calI^+) - v_i(\calI^-)| \leq |\mu_i(\rplus) - \mu_i(\rminus)| + \eps' \leq (\eps - \eps') + \eps' = \eps$, so this is indeed a solution to \eps-\ch.

        Again, notice that the construction we described can be done in polynomial time, since the creation of each big-tile can be done in polynomial time, and we have $1/d$ many big-tiles, where $1/d = \ceil{4 \cdot n \cdot c_{\max} / \eps'} \leq \ceil{4 \cdot n \cdot c_{\max} \cdot n^r} $ for some constant $r \geq 1$.
	\end{proof}
	
	The above, together with the fact that \eps-\ch is \ppa/-hard for any $\eps < 1/5$ (\cite{DFHM25}) implies the main theorem of this section.
	
	\begin{theorem}\label{thm:hvu-pizza-ppa-h}
		\eps-\scut-\pizza with $n$ mass distributions is \ppa/-hard for any constant $\eps < 1/5$, even when $n-1+n^{1-\delta}$ turns are allowed in the \scut-path for any given constant $\delta \in (0,1]$, every mass distribution is uniform over polynomially many rectangles, and there is no overlap between any two mass distributions.
	\end{theorem}

        Similarly to the case of the decision variant of \eps-\spizza (\cref{thm:straight-np-h-overlap}), we can get the following result by reducing from the \np/-hard instances of \eps-\ch for constant \eps.

	\begin{theorem}\label{thm:square-np-h-overlap}
		There exists a constant $\eps > 0$ for which it is \np/-hard to decide if an \eps-\scut-\pizza instance with $n$ mass distributions admits a solution consisting of a \scut-path with at most $n-2$ turns, even when every mass distribution is uniform over polynomially many rectangles, and there is no overlap between any two mass distributions.
	\end{theorem}

	\subsection{Hardness of discrete \spizza and \scut-\pizza} \label{sec: discrete-pizza_hardness}
	
	In this section, we study the discrete versions of \spizza and \scut-\pizza.
	
	\begin{mdframed}[backgroundcolor=white!90!gray,
		leftmargin=\dimexpr\leftmargin-20pt\relax,
		innerleftmargin=4pt,
		innertopmargin=4pt,
		skipabove=5pt,skipbelow=5pt]
		\begin{definition}
			\label{def:discrete-straight-pizza} For any $n \geq 1$, the problem \eps-\dspizza is defined as follows:
			\begin{itemize}
				\item \textbf{Input:} $\eps \geq 0$, and $2n$ point sets $P_1, P_2, \dots, P_{2n}$ on $[0,1]^2$.
				\item \textbf{Output:} One of the following.
				\begin{enumerate}
					\item[(a)]\label{discr-straight-pizza-point:a} Three points that can be intersected by the same line. 
					\item[(b)] A partition of $[0,1]^2$ to \rplus and \rminus using at most $n$ lines such that for each $i \in [2n]$ it holds that
                    $\left| |P_i \cap \rplus| - |P_i \cap \rminus| \right| \leq \eps \cdot |P_i|$.
				\end{enumerate} 
			\end{itemize}
			A point that is intersected by a line does not belong to any of \rplus, \rminus.
		\end{definition}
	\end{mdframed} \vspace{5pt}

	\begin{mdframed}[backgroundcolor=white!90!gray,
		leftmargin=\dimexpr\leftmargin-20pt\relax,
		innerleftmargin=4pt,
		innertopmargin=4pt,
		skipabove=5pt,skipbelow=5pt]
		\begin{definition}
			\label{def:discrete-sc-pizza} For any $n \geq 1$, the problem \eps-\dpizza is defined as follows:
			\begin{itemize}
				\item \textbf{Input:} $\eps \geq 0$, and $n$ point sets $P_1, P_2, \dots, P_{n}$ on $[0,1]^2$.
				\item \textbf{Output:} One of the following. 
				\begin{enumerate}
					\item[(a)]\label{discr-sc-pizza-point:a} Two points with the same $x$- or $y$-coordinate.
					\item[(b)] A partition of $[0,1]^2$ to \rplus and \rminus using a $y$-monotone \scut-path with at most $n-1$ turns such that for each $i \in [n]$ it holds that $\left| |P_i \cap \rplus| - |P_i \cap \rminus| \right| \leq \eps \cdot |P_i|$.
				\end{enumerate} 
			\end{itemize}
			A point that is intersected by a line does not belong to any of \rplus, \rminus.
		\end{definition}
	\end{mdframed} \vspace{5pt}

	Notice that the first kind of allowed output for both problems (\cref{def:discrete-straight-pizza}(a), \cref{def:discrete-sc-pizza}(a)) is a witness that the input points are not in general position or that their $x$- or $y$-coordinates are not unique, respectively, which can be checked in polynomial time. The second kind of output (\cref{def:discrete-straight-pizza}(b), \cref{def:discrete-sc-pizza}(b)) is the one that is interesting and can encode the hard instances studied here. In case the first kind of output does not exist, the other one is guaranteed to exist due to \cite{S21} for \eps-\dspizza, while for \eps-\dpizza its existence is guaranteed for every $\eps \in [0,1]$ due to a reduction we present in \cref{sec: discrete-pizza_containment} which shows \ppa/ membership.
	
	The definition of \eps-\dspizza is a slightly modified form of the one that appears in \cite{S21}, where it is referred to as \textsc{DiscretePizzaCutting}. In our definition, we avoid having to ``promise'' an input of points that are in general position, by allowing as output a witness of an inappropriate input. Furthermore, the definition we present is more general, since it accommodates an approximation factor $\eps$; in particular, for $\eps < \min_{i} \{ 1/|P_i| \}$ we get the definition of the aforementioned paper. Similarly, we define \eps-\dpizza, which to the best of our knowledge, has not been stated in previous work.
	Note that, if the input of any of the discrete versions of the problems consists of points that are in general position and furthermore have pairwise different $x$- and $y$-coordinates (a property that the instances in our reductions have), in any \eps-\dspizza solution, a line can only intersect up to two points, while in any \eps-\dpizza solution, a line segment can intersect up to one point.
	
	Recall that in \eps-\scut-\pizza we are given $n$ mass distributions, while in \eps-\spizza we are given $2n$ mass distributions as input. In \cref{app: cont-to-discr}, we describe a general construction that takes as an input $q \in \{n, 2n\}$ mass distributions $\mu_1, \dots, \mu_q$ normalized on $[0,1]^2$, represented by weighted polygons with holes (see \cref{sec:preliminaries} for the detailed description), and turns it into $q$ sets of points $P_1, \dots, P_q$ on $[0,1]^2$. The points that constitute those sets' union are in general position, and additionally, they have unique $x$- and $y$-coordinates. We prove that, if a set of at most $n$ lines or a \scut-path of at most $n-1$ turns partitions $[0,1]^2$ into $\rplus$ and $\rminus$ such that $\left| |P_i \cap \rplus| - |P_i \cap \rminus| \right| \leq (\eps - \eps') \cdot |P_i|$, then the same set of lines or \scut-path, respectively, separates the mass distributions such that $\left| \mu_{i}(\rplus) - \mu_{i}(\rminus) \right| \leq \eps$.
	
	Let $N \geq 2n$ be the input size of any of our two pizza sharing problems, and let the smallest area triangle in the mass distributions' triangulation be $\alpha$. For any $\eps'<\eps$, where $\eps$ and $\alpha$ are at least inverse polynomial in $N$, the construction results in a polynomial time reduction from \eps-\spizza to $(\eps-\eps')$-\dspizza and from \eps-\scut-\pizza to $(\eps-\eps')$-\dpizza.
	The reduction can be performed in time polynomial in the input size and in $1/\alpha$. It consists of two parts: (i) First we ``pixelate'' the mass distributions finely enough so that they are represented by a sufficiently large number of pixels. This will ensure a high enough ``resolution'' of the pixelated distributions. (ii) The pixels will then be turned into points, which we have to perturb in order to guarantee they are in general position and with unique coordinates, as required. In particular, in \cref{app: cont-to-discr}, we prove the following.
	\begin{lemma}\label{lem: cont-to-discr}
		Let $N$ be the input size of an approximate pizza sharing problem (either \eps-\spizza or \eps-\scut-\pizza) whose triangulation has no triangle with area less than $\alpha > 0$. Also, let $\eps' \in \left[\frac{6}{N^c}, \eps \right)$, where $c > 0$ is a fixed constant, and $\frac{6}{N^c} < \eps < 1$. Then, the instance can be reduced in time $\poly(N, 1/\alpha)$ to its approximate discrete version, that is, $(\eps - \eps')$-\dspizza or $(\eps - \eps')$-\dpizza, respectively.   
	\end{lemma}
	
	Given \cref{thm:straight-pizza-PPA} and \cref{thm:hvu-pizza-ppa-h}, and since their instances are constructed such that $\alpha$ is an at least inverse polynomial function of the input size, \cref{lem: cont-to-discr} implies the following hardness results.
	
	\begin{theorem}\label{thm: discr-straight-pizza}
		\eps-\dspizza with $2n$ point sets is \ppa/-hard for any constant $\eps < 1/5$, even when $n+n^{1-\delta}$ lines are allowed for any given constant $\delta \in (0,1]$.
	\end{theorem}

	\begin{theorem}\label{thm: discr-sc-pizza}
		\eps-\dpizza with $n$ point sets is \ppa/-hard for any constant $\eps < 1/5$, even when $n-1+n^{1-\delta}$ turns are allowed in the \scut-path for any given constant $\delta \in (0,1]$.
	\end{theorem}

	We note that \ppa/-hardness for \eps-\dspizza was so far known only for any $\eps \in \left[ 0, \min_{i} \{ 1/|P_i| \} \right)$ (which is equivalent to $\eps = 0$), due to \cite{S21}. Since, in the aforementioned paper's constructions, $\min_{i} \{ 1/|P_i| \} \in O(1/\poly(N))$, our result strengthens the hardness of the problem significantly.

        \cref{lem: cont-to-discr} is general enough to allow us to derive \np/-hardness results for the decision variants of the two discrete versions of the pizza sharing problems. If we ask for a solution with at most $n-1$ straight lines or $n-2$ turns in $(\eps - \eps')$-\dspizza and $(\eps - \eps')$-\dpizza, respectively, then we can easily reduce to them from the instances of \cref{thm:straight-np-h-overlap} and \cref{thm:square-np-h-overlap}, picking $\eps'$ to be some inverse polynomial function of $N$, e.g., $\eps' = 1/N$. In particular, we get the following.

        \begin{theorem}\label{thm: discr-straight-pizza-np-hard}
		There exists a constant $\eps > 0$ for which it is \np/-hard to decide whether a solution of \eps-\dspizza with $2n$ point sets and at most $n-1$ lines exists.
	\end{theorem}

	\begin{theorem}\label{thm: discr-sc-pizza-np-hard}
		There exists a constant $\eps > 0$ for which it is \np/-hard to decide whether a solution of \eps-\dpizza with $n$ point sets and a \scut-path with at most $n-2$ turns exists.
	\end{theorem}

	\subsection{Hardness of exact \scut-\pizza}
	\label{sec:exact-pizza_hardness}
	
	In this section, we show hardness results for \emph{exact} \scut-\pizza, that is, \cref{def:eps-hv-pizza} for $\eps = 0$. We prove
	that solving \scut-\pizza is \fixp/-hard and that deciding whether there exists a
	solution for \scut-\pizza  with fewer than $n-1$ turns is \etr/-hard.
	
	As mentioned earlier, computing an exact solution of a \fixp/-hard problem may require computing an irrational number. To showcase this for \scut-\pizza, consider the following simple instance. Let us have a single mass distribution in the shape of a right-angled triangle, whose corners are on $(0,1)$, $(1,1)$, and $(1,0)$. It is normalised, i.e., its total mass is $1$, therefore its weight is $2$. An exact solution of \scut-\pizza is either a horizontal or a vertical straight line (with $0$ turns), which cuts the triangle such that each half-space has half of the mass, that is, $1/2$. One can easily check that the solution is either the horizontal line $y=\sqrt{2}/2$, or the vertical line $x=\sqrt{2}/2$.
	
	We provide a main reduction from (exact) \ch to (exact) \scut-\pizza, whose gadgets are then employed to show the \etr/-hardness of the decision variant. This time, as a starting point, we will use the instances of \ch constructed as end-points of the reductions in \cite{DFMS21}. In particular, when clear from context, \dfms will denote an arbitrary instance satisfying properties \ref{point:1}, \ref{point:2}, and \ref{point:3} (see next paragraph). When requesting a solution of \dfms with at most $n$ cuts, we get \fixp/-hardness, while when requesting to decide if it is solvable with $n-1$ cuts we get \etr/-hardness. Both of these results are due to \cite{DFMS21}, and hold even for $6$-block-triangle valuations. We note that here the input consists of sets of points with rational coordinates, i.e., we describe polygons by their vertices. For a detailed description of the input representation, see \cref{sec:preliminaries}. In the aforementioned work, the problems' input is $n$ algebraic circuits capturing the cumulative valuation of $n$ agents on $[0,1]$, which are piece-wise polynomials of maximum degree 2. In particular, since their (density) valuation functions are piece-wise linear, the input of \scut-\pizza suffices to consist of very simple shapes, namely, only rectangles and triangles.

	\myparagraph{The reduction.} 
	Here we show the main reduction, which will conclude with the proof that finding an exact solution to \scut-\pizza is \fixp/-hard. Then, this reduction's construction will be used to show that the problem's decision variant is \etr/-hard. We reduce from a \ch instance \dfms with $n$ agents and $6$-block-triangle valuations to a \scut-\pizza instance \isc with $n$ mass distributions.  

	The key difference between this reduction and the previous reductions on the approximate versions is that
	the starting point of the reduction, i.e., instance \dfms, besides rectangular-shaped (constant) pieces, also contains triangular-shaped (linear) pieces in the valuation density functions for some agents. 
	More specifically, all of the following properties hold for \dfms:
	\begin{enumerate}
		\item the valuation function of every agent is 4-block-triangle, or 6-block; \label{point:1} 
		\item for any given agent $i \in [n]$, every triangle (linear piece) of her valuation function has height 2 and belongs to exactly one interval of interest of the form $[x_{j}, x_{j+1}]$ for $j \in [m]$, where $m \leq 12n + 1$ (see below for the definition of those intervals); \label{point:2}
		\item for every agent $i \in [n]$ there exists an interval $[a_i, b_i]$ that contains more than half of their total valuation (i.e., more than $1/2$ cumulative valuation), and in addition, for every $i' \neq i$ we have $(a_i,b_i) \cap (a_{i'},b_{i'}) = \emptyset$. \label{point:3} 
	\end{enumerate}
	Also, in this reduction, the resulting \scut-\pizza instances will contain \emph{weighted} mass distributions (see definition in \cref{sec:preliminaries}).

	The first
	step is to partition $[0,1]$ of \dfms into subintervals that are defined
	by \emph{points of interest}.  We say that a point $x \in [0,1]$ is a point
	of interest if it coincides with the beginning or the end of a valuation block or triangle
	of an agent;  formally, $x$ is a point of interest if $x \in \{a^\ell_{ij},
	a^r_{ij}\}$ for some agent $i \in [n]$ and some valuation block $j$ or triangle of hers (for this notation see \cref{sec:preliminaries}). These points
	conceptually split $[0,1]$ into \emph{intervals of interest}, since in between any
	pair of consecutive points of interest, all agents have a non-changing
	valuation density. Let $0 =: x_1 < x_2 <
	\ldots < x_{m} < x_{m+1} := 1$ denote the points of interest, and each interval of interest $[x_{j},x_{j+1}]$, $j \in [m]$ is called the \emph{$j$-th subinterval}. Observe that $m \leq 12 n + 1$: as a base case consider a single block or triangle produces at most two points of interest and therefore at most three intervals, so $m \leq 3$; for any block or triangle we add, we increase by at most 2 the points of interest; each valuation function has either $4$ blocks and $1$ triangle, or $6$ blocks, therefore the total number of blocks and triangles is at most $n \cdot 6$ (by Property \ref{point:1} above). 

    Here it is important to mention that in our reduction we are allowed to only use \fixp/ gates (see \cref{sec:preliminaries} for details). As representation of the \dfms instance,\footnote{This is not the same input representation as the one defined in \cite{DFMS21} for exact \ch. However, it is easy to check that the \fixp/-hardness reduction of the aforementioned work goes through if we require this new input representation of \ch. } for each agent $i \in [n]$ we consider $\ell_{i}$ ordered pairs of points in $[0,1]$ interpreted as consecutive intervals' endpoints, together with their corresponding valuation density function in the form of a circuit: $\left( (r_{k}^{(i)}, r_{k+1}^{(i)}), f_{k}^{(i)} (x) \right)_{k \in [\ell_{i}]}$, where $r_{1}^{(i)} :=0$ and $r_{\ell_i + 1}^{(i)} := 1$ for all $i \in [n]$. In particular, according to \dfms, for any $i \in [n]$, $k \in [\ell_{i}]$, either $f_{k}^{(i)} (x) = c_{ik}$ or $f_{k}^{(i)} (x) = 2(x-r_{k}^{(i)})c_{ik}$, where $c_{ik} \geq 0$ is a constant.
    
    The next step is to do some preprocessing of the input in order to incorporate the intervals of interest. The points of interest can be identified using an algebraic circuit by a sorting network (e.g., \cite{Knuth98-sorting}) which takes as input all points $\left( r_{k}^{(i)} \right)_{i \in [n], k \in [\ell_i]}$ and outputs them in non-decreasing order $\left( x_{j} \right)_{j \in [m+1]}$. Then, we turn the \dfms instance representation into the following form for each agent $i \in [n]$: $\left( (x_{j,k}^{(i)}, x_{j+1,k}^{(i)}), f_{k}^{(i)} (x) \right)_{k \in [\ell_i], j \in [m]}$, where
    \begin{align*}
        x_{j,k}^{(i)} = \max \{ r_{k}^{(i)}, \min \{ x_{j}, r_{k+1}^{(i)} \} \}.
    \end{align*}
    Observe that, by definition, for any $i \in [n]$, $k \in [\ell_i]$, and $j \in [m]$, either $[x_{j}, x_{j+1}] \cap [r_{k}^{(i)}, r_{k+1}^{(i)}] = [x_{j}, x_{j+1}]$, or  $[x_{j}, x_{j+1}] \cap [r_{k}^{(i)}, r_{k+1}^{(i)}]$ is singleton or empty. So, we have the following cases:
    \begin{enumerate}
        \item[(a)] $x_{j} \leq r_{k}^{(i)}$ (and $x_{j+1} \leq r_{k}^{(i)}$): The first inequality implies $x_{j} \leq r_{k+1}^{(i)}$ and so, $x_{j,k}^{(i)} = r_{k}^{(i)}$. The second inequality implies $x_{j+1} \leq r_{k+1}^{(i)}$), and so, $x_{j+1,k}^{(i)} = r_{k}^{(i)}$. Therefore, $x_{j,k}^{(i)} = x_{j+1,k}^{(i)} = r_{k}^{(i)}$.
        \item[(b)] $x_{j} \geq r_{k+1}^{(i)}$ (and $x_{j+1} \geq r_{k+1}^{(i)}$): Similarly to above case, $x_{j,k}^{(i)} = x_{j+1,k}^{(i)} = r_{k+1}^{(i)}$.
        \item[(c)] $x_{j} \geq r_{k}^{(i)}$ and $x_{j+1} \leq r_{k+1}^{(i)}$: By definition, $x_{j} \leq x_{j+1}$, so we get $x_{j,k}^{(i)} = x_{j}$ and  $x_{j+1,k}^{(i)} = x_{j+1}$.
    \end{enumerate}
    This way, using \fixp/ gates we managed to copy the valuation density function $f_{k}^{(i)} (x)$ to all the intervals $[x_{j} , x_{j+1}]$ that are inside $[ r_{k}^{(i)} , r_{k+1}^{(i)} ]$, while for those that are outside of it, we created artificial singleton intervals which do not contribute to the cumulative valuation function (since $f_{k}^{(i)}$ is a density function).	
	
    Now we are ready to construct the \scut-\pizza instance \isc.
	For each subinterval $j \in [m]$ of \dfms, we will construct a tile $t_j$ of size $d_j \times d_j$ in which we will place our measures, where $d_j := x_{j+1} - x_{j}$, and notice that always $d_j > 0$. These tiles will be placed diagonally in $[0,1]^2$, starting from the bottom-left corner (see \cref{fig:triangles-b} for a depiction). The points describing $t_j$ are $LL := \left( x_{j}, x_{j} \right)$, $HL := \left( x_{j+1}, x_{j} \right)$, $HH := \left( x_{j+1}, x_{j+1} \right)$, and $LH := \left( x_{j}, x_{j+1} \right)$. Inside each tile $t_j$, we place triangles $z_{j,k}^{(i)}$ for all $i \in [n], k \in [\ell_i]$, each with vertices $HL$, $HH$, and $LH$. Triangle $z_{j,k}^{(i)}$ has weight $w_{j,k}^{(i)} = d_{j}^{-2} \cdot (x_{j+1,k}^{(i)} - x_{j,k}^{(i)}) \cdot f_{k}^{(i)} (x_{j+1,k}^{(i)})$. Notice that here we have also used a division gate.\footnote{The division gate is among those available in \fixp/ and can be used as long as there is no division with 0. However, it is known that the definition of the class does not need it since there are \fixp/-hard problems that do not use this gate. For a proof of this, see \cite{EY10}.} Also, we place triangles $z_{j,k}^{*(i)}$ for all $i \in [n], k \in [\ell_i]$, each with vertices $LL$, $HL$, and $HH$. Triangle $z_{j,k}^{*(i)}$ has weight $w_{j,k}^{*(i)} = d_{j}^{-2} \cdot (x_{j+1,k}^{(i)} - x_{j,k}^{(i)}) \cdot f_{k}^{(i)} (x_{j,k}^{(i)})$ (notice the change in the argument of $f_{k}^{(i)}$).

    By the above construction, if $f_{k}^{(i)} (x) = c_{ik}$, then $w_{j,k}^{(i)} = w_{j,k}^{*(i)} =  d_{j}^{-2} \cdot (x_{j+1,k}^{(i)} - x_{j,k}^{(i)}) \cdot c_{ik}$, and by multiplying all parts with $d_{j}^{2}$ we get that for every colour $i \in [n]$ the area of the $j$-th subinterval equals its measure inside tile $t_j$, and has square shape. Similarly, if $f_{k}^{(i)} (x) = 2 (x - r_{k}^{(i)}) c_{ik}$, then $w_{j,k}^{(i)} = d_{j}^{-2} \cdot (x_{j+1,k}^{(i)} - x_{j,k}^{(i)}) \cdot 2 (r_{k+1}^{(i)} - r_{k}^{(i)}) c_{ik}$, and $w_{j,k}^{*(i)} = 0$. By multiplying both sides of the former equation with $d_{j}^{2} / 2$ we get that for every colour $i \in [n]$ the area of the $j$-th subinterval equals its measure inside tile $t_j$, and has triangular shape (see \cref{fig:triangles}).
 
	\begin{figure}[htbp]
		\begin{subfigure}[b]{0.48\textwidth}
			\centering
			\includegraphics[width=0.75\linewidth]{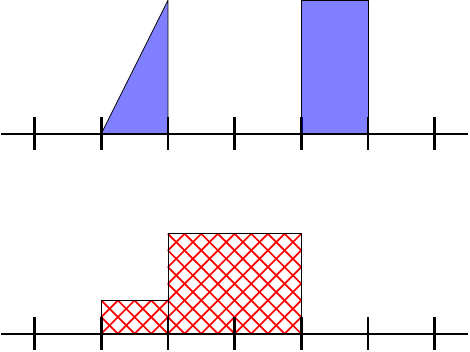}
			\caption{Part of the \ch instance with two agents and the corresponding regions of interest.} \label{fig:triangles-a}
		\end{subfigure}%
		\hspace*{\fill}   
		\begin{subfigure}[b]{0.47\textwidth}
			\centering
			\includegraphics[width=0.65\linewidth]{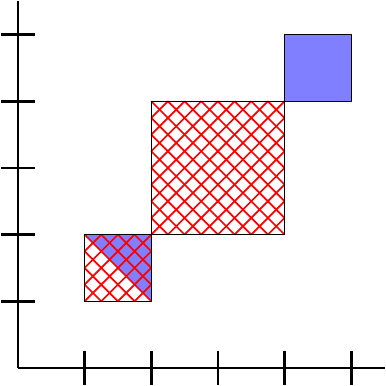}
			\caption{The corresponding part of the \scut-\pizza instance. \\ } \label{fig:triangles-b}
		\end{subfigure}%
		\caption{An example of the reduction from \ch to \scut-\pizza.}
		\label{fig:triangles}
	\end{figure}

	Now we need to show how an exact solution of \isc, that is, a \scut-path with $n-1$ many turns is mapped back to a solution of \dfms with $n$ cuts. This is straightforwardly done in the following way. Let the solution be represented by an ordered tuple of $(n-1)$ points $(p_1, \dots, p_{n-1})$ interpreted as the turns of the \scut-path. Then, $p_1$'s coordinates correspond to the first two cuts in $[0,1]$ of \dfms, and for the remaining $n-2$ points, those with even index encode a cut at their $y$-coordinate, while those with odd index encode it at their $x$-coordinate.
	
    It remains to prove that this is a solution of \dfms. 
    To do this, we will use the following crucial observation.
	
	\begin{claim}
		\label{clm:no-turns-triangles}
		In any solution of \isc created by \dfms, there can be no turn inside a tile.
	\end{claim}
	
	\begin{proof}
		The truth of the statement can become apparent if one considers Property \ref{point:3} of the above facts on \dfms, together with the fact that the endpoints of each $[a_i , b_i]$ are points of interest. It is implied then, that in any of the \dfms solutions, there needs to be at least one cut in each interval $[a_i, b_i]$ for each $i \in [n]$. And since we are allowed to draw at most $n$ cuts, there will be a single cut in each of those intervals. Also, due to the fact that $a_i, b_i$ are points of interest, each cut in $[a_i , b_i]$ belongs to a different subinterval, and therefore, there will be exactly $n$ cuts in $n$ distinct subintervals. Focusing now on our \isc construction, the \scut-path with $n-1$ turns consists of a total of $n$ horizontal and vertical line segments. If any of those does not intersect any tile, then this \scut-path will correspond to a set of at most $n-1$ cuts in \dfms, which cannot be a solution. Therefore, every line segment intersects some tile. 
		
		Notice that, due to the diagonal placement of tiles, any \scut-path that is a solution has to be positively $x$-monotone and positively $y$-monotone, i.e., to have a ``staircase'' form. Also, this diagonal placement of tiles dictates that, if there was a turn of the \scut inside a tile, then two line segments are used to intersect it instead of one. This means that at most $n-1$ tiles will be intersected, and therefore this translates to a set of at most $n-1$ cuts in \dfms, which cannot be a solution.
	\end{proof}
	Put differently, having a turn inside a tile would be a ``waste'', and in our instances, all turns are needed in order for a solution to exist. Now, if we focus on any tile $t_j$ for some $j \in [m]$, as shown in the construction, it will contain exactly the same measure that the corresponding colour has inside the $j$-th subinterval of \dfms. Furthermore, if the density is rectangular then the measure in the tile is square. Having a square on the entire region of the tile allows us to immediately translate a line segment of \scut-path that intersects the square into a \dfms cut, since we translate in the same manner both horizontal and vertical such segments. Similarly, if the density is triangular, then so is the measure in the tile half of whose area it occupies, and the right angle of the triangle is at the top-right of the tile. This again allows us to translate any line segment, horizontal or vertical, straightforwardly into a cut of \dfms, since the measure at the bottom or left, respectively, part of the triangle has exactly the same area as that of the corresponding triangle's part at the left of the \dfms cut. 
	
	The above analysis shows that, for any $i \in [n]$, in each individual tile the total \lplus mass is equal to the total \lplus value of the corresponding subinterval. Therefore, given a solution to \isc where the total mass of \rplus will equal that of \rminus for every $i \in [n]$, the induced cuts on \ich constitute a solution. Finally, by construction of \dfms, at any given point in $[0,1]$, no more than three agents have positive valuation density, therefore, \isc has overlap at most $3$. Due to the \fixp/-hardness of exact \ch shown in \cite{DFMS21}, we get the following.
    \begin{theorem}\label{thm:fixp-h}
		\scut-\pizza is \fixp/-hard even when every mass distribution consists of at most six pieces that can be unit-squares or right-angled triangles, and have overlap at most 3.
    \end{theorem}

	We can also show that deciding whether there exists an exact \scut-\pizza solution with $n-2$ turns is \etr/-hard. For this, we will use a result of \cite{DFMS21}, where it was shown that deciding whether there exists an exact \ch solution with $n$ agents and $n-1$ cuts is \etr/-hard. We give a reduction from this version of \ch to the decision problem for \scut-\pizza. The reduction uses the same ideas that we presented for the \fixp/-hardness reduction.

	Before we prove the theorem, let us give a brief sketch of the \etr/-hardness proof of the aforementioned \ch decision variant of \cite{DFMS21}. We are given an instance of the following problem which was shown to be \etr/-complete (Lemma 15 of the aforementioned paper).
	\begin{mdframed}[backgroundcolor=white!90!gray,
		leftmargin=\dimexpr\leftmargin-20pt\relax,
		innerleftmargin=4pt,
		innertopmargin=4pt,
		skipabove=5pt,skipbelow=5pt]
		\begin{definition}[\bfeas]
			\label{def:feas}
			Let $p(x_1, \dots, x_m)$ be a polynomial. 
			We ask whether there exists a point $(x_1, \dots, x_m) \in [0,1]^m$ that satisfies $p(x_1, \dots, x_m) = 0$.
		\end{definition}
	\end{mdframed} \vspace{5pt}

	Given the polynomial $p$, we first normalize it so that the sum of the absolute values of its terms is in $[0,1]$ (thus not inserting more roots), resulting in a polynomial $q$. Then, we separate the terms that have positive coefficients from those that have negative coefficients, thus creating two positive polynomials $q_1, q_2$ such that $q = q_1 - q_2$. Therefore, $p(\vec{x})=0$ for some $\vec{x} \in [0,1]^{m}$ if and only if $q_{1}(\vec{x}) = q_{2}(\vec{x})$. We then represent $q_1, q_2$ in a circuit form with gates that implement the operations $\{ c, +, \times c, \times \}$ (where $c \in [0,1] \cap \mathbb{Q}$ is a constant input, and $\times c$ is multiplication by constant). By the scaling we know that $q_1, q_2 \in [0,1]$, and in addition, the computation of the circuit using the aforementioned operations can be simulated by a \ch instance with $n-1$ agents, where $n-1 \in \poly(\# gates)$; the argument of $p$ becomes a set of ``input'' cuts and according to the circuit implementation by \ch, two output cuts encode $q_1$ and $q_2$. Finally, checking whether $q_1 = q_2$ is true is done by an additional $n$-th agent that can only be satisfied (have her total valuation split in half) if and only if $q_1(\vec{x}) = q_2(\vec{x})$ for some $\vec{x} \in [0,1]^{m}$. In other words, the \ch instance has a solution if and only if there are ``input'' cuts $\vec{x} = (x_1, \dots, x_m) \in [0,1]^m$ that force the rest of the cuts (according to the circuit implementation) that encode the values $v_{m+1}, \dots, v_{n-1}$ at the output of each of the circuit's gates such that they also cut the $n$-th agent's valuation in half without the need for an additional cut.    
	
    We are now ready to prove the following theorem.
    
        \begin{theorem}\label{thm:etr-h}
            It is \etr/-hard to decide if an exact \scut-\pizza instance admits a solution with a \scut-path with at most $n-2$ turns, even when every mass distribution consists of at most six pieces that can be unit-squares or right-angled triangles, and have overlap at most 3.
    \end{theorem}
    
    \begin{proof}
        We will use exactly the aforementioned technique from \cite{DFMS21} up to the point where we have a \ch instance that checks whether $q_1 = q_2$. Then, we use the gadgets described in the \fixp/-hardness reduction (proof of \cref{thm:fixp-h}) that reduce the valuation functions of $n$ agents in \ch into mass distributions of a \scut-\pizza instance with $n$ colours. According to \cref{clm:no-turns-triangles}, in any solution of the resulting \scut-\pizza instance, the \scut-path does not have turns inside any unit-square. This means that each horizontal/vertical segment of the \scut-path that cuts a unit square in \scut-\pizza has a 1-1 correspondence to a cut of a \ch solution, thus a \ch solution that uses $n-1$ cuts would correspond to a \scut-path with $n-1$ line segments, i.e., $n-2$ turns. Therefore, if and only if there is a \scut-path that solves \scut-\pizza with $n$ colours and $n-2$ turns, there is a $(n-1)$-cut that solves \ch with $n$ agents. Equivalently, there is a $\vec{x} \in [0,1]^m$ such that $p(\vec{x})=0$, making the \bfeas instance satisfiable. 
    \end{proof}

	\section{Membership results}
	\label{sec: existence}

    Up to this point, we have showed that \spizza and \scut-\pizza are \ppa/-hard for their approximate versions for any discrepancy $\eps < 1/5$, even when the input consists of point sets. We have also showed that the decision variants of the problems are NP-hard. Furthermore, we have studied the exact version of \scut-\pizza and proved that it is \fixp/-hard, while its decision variant is \etr/-hard.
 
	In this section, we present membership results for the exact and approximate versions of the aforementioned pizza sharing problems. 
	Our \ppa/ membership result for \spizza is achieved by reducing the problem to its discrete version, which was recently shown by \cite{S21} to be in \ppa/ (\cref{thm: straight-pizza-PPA-incl}). For \scut-\pizza, our results revolve around the original existence proof of \cite{KRS16} which, additionally to the Borsuk-Ulam theorem, uses other involved topological techniques. We present a new algorithmic proof based on the original one, where now the only topology tool used is the Borsuk-Ulam theorem (\cref{thm: sc-pizza_exists}). Then, by showing how to algorithmically compute the measure contained in the positive part of an arbitrary \scut-path (\cref{app: exact-computation-measures}), we turn this into a \ppa/ membership proof (\cref{thm: eps-sc-pizza-in_ppa}). Finally, we study the corresponding decision variants of the problems and acquire \np/ membership for the approximate continuous and discrete versions of the problems and \etr/ membership of exact \scut-\pizza.

	\subsection{Membership results for approximate \spizza}
	\label{sec:straight-pizza_inclusion}
	
	Here we prove that \eps-\spizza with $2n$ mass distributions is in \ppa/ for any $\eps \in \Omega(1/\poly(N))$ and $\alpha \in \Omega(1/\poly(N))$, where $N$ is the input size and $\alpha$ is the smallest area among the triangles of the triangulated mass distributions of \eps-\spizza. This answers a big open question left from \cite{DFM22}, where \ppa/ membership was elusive. We also show that deciding whether a solution with at most $n-1$ straight lines exists is in \np/. Both those results are derived by reducing the problem to its discrete version, namely \dspizza, where instead of mass distributions, the input consists of points, and the goal is to bisect (up to one point) each of the $2n$ point sets using at most $n$ straight lines. \dspizza, was recently shown by \cite{S21} to be in \ppa/.
	
    \myparagraph{\ppa/ membership.}
    We will use \cref{lem: cont-to-discr} in a straightforward way. In particular, given an \eps-\spizza instance for some $\eps \in \Omega(1/\poly(N))$, we will pick an $\eps' < \eps$ as prescribed in the aforementioned lemma, and reduce our problem to \dspizza in time $\poly(N,1/\alpha)$.

        \begin{theorem}
		\label{thm: straight-pizza-PPA-incl}
		\eps-\spizza with $2n$ weighted mass distributions with holes is in \ppa/ for any $\eps \in \Omega(1/\poly(N))$ and $\alpha \in \Omega(1/\poly(N))$, where $N$ is the input size and $\alpha$ is the smallest area among the triangles of the triangulated mass distributions.
	\end{theorem}

	\myparagraph{\np/ membership.}
        Observe that \cref{lem: cont-to-discr} shows how to turn a set of $2n$ mass distributions (weighted polygons with holes) into a set of (unweighted) points, which, if cut by at most $2n$ straight lines, will result to an approximate cut of the mass distributions relaxed by an extra additive $\eps' \in \Omega (1/N^c)$ for any $c > 0$. When the number of straight lines is at most $n-1$, this gives straightforwardly a reduction to $0$-\dspizza, since we can check how many of the polynomially many points of the latter instance are in \rplus and \rminus.

        \begin{theorem}\label{thm: straight-pizza-np-contain}
            Deciding whether \eps-\spizza with $2n$ weighted mass distributions with holes has a solution with at most $n-1$ straight lines is in \np/ for any $\eps \in \Omega(1/\poly(N))$ and $\alpha \in \Omega(1/\poly(N))$, where $N$ is the input size and $\alpha$ is the smallest area among the triangles of the triangulated mass distributions.
        \end{theorem}

	\subsection{Membership results for approximate \scut-\pizza} \label{sec: approx_inclusion}

        Here we show that the problem of finding a solution to \eps-\scut-\pizza is in \ppa/ even for exponentially small \eps, while deciding whether there exists a solution (a \scut-path) with at most $k \in \naturals$ turns is in \np/. The latter is almost immediate by the fact that any candidate solution is verifiable in polynomial time. The former is shown by reducing \eps-\scut-\pizza to \epsborsuk which is in \ppa/ (e.g., see \cite{DFMS21}).

    \myparagraph{Existence of a \scut-\pizza solution.} 
	We begin by proving that a solution to exact \scut-\pizza always exists (and therefore, a solution to the approximate version exists too).
	This proof holds for arbitrary mass distributions, but for our algorithmic results, we will only consider the case where the mass distributions are unions of polygons with holes. Our proof is based on that of Karasev, Rold{\'a}n-Pensado, and Sober{\'o}n \cite{KRS16}. However, they use more involved techniques from topology, which we would like to avoid since our goal is to give an algorithmic proof.
	
    \begin{theorem}[originally by \cite{KRS16}] \label{thm: sc-pizza_exists}
		Let $n$ be a positive integer. For any $n$ mass distributions in $\reals^2$, there is a path formed by only horizontal and vertical segments with at most $n-1$ turns that splits $\reals^2$ into two sets of equal size in each measure. Moreover, the path is $y$-monotone.
    \end{theorem}

    \begin{proof}
	Let $S^n$ denote the $L_1$ sphere in $n+1$ dimensions. The Borsuk-Ulam theorem states that if $f : S^n \to \reals^n$ is a continuous function, then there exists a point $\vec{z} \in S^n$ such that $f(\vec{z}) = f(-\vec{z})$. Consider $n$ mass distributions in $\reals^2$.
	For some given $d \in \naturals^*$, we will show how to decode \scut-paths from points in $S^d$, and then we will construct a function $f: S^d \to \reals^n$ for which $f(\vec{z}) = f(-\vec{z})$ implies that the \scut-path corresponding to $\vec{z}$ is a solution to the \scut-\pizza problem. In particular, we will show a general way to decode \scut-paths (with $d-1$ turns) from points in $S^d$ for a suitable dimension $d \in \naturals^*$ to be determined later. It will turn out that, to guarantee a \scut-\pizza solution given $n$ measures, the Borsuk-Ulam theorem will require $d = n$. However, we will have $d$ undetermined for as long as we can, making the proof transparent enough to help in the understanding of the cases where $d \neq n$ (see \cref{thm: NP-incl}, \cref{thm: ETR-incl}, and their proofs).

 	\begin{figure}[htbp]
        \centering
        \includegraphics[scale=0.67]{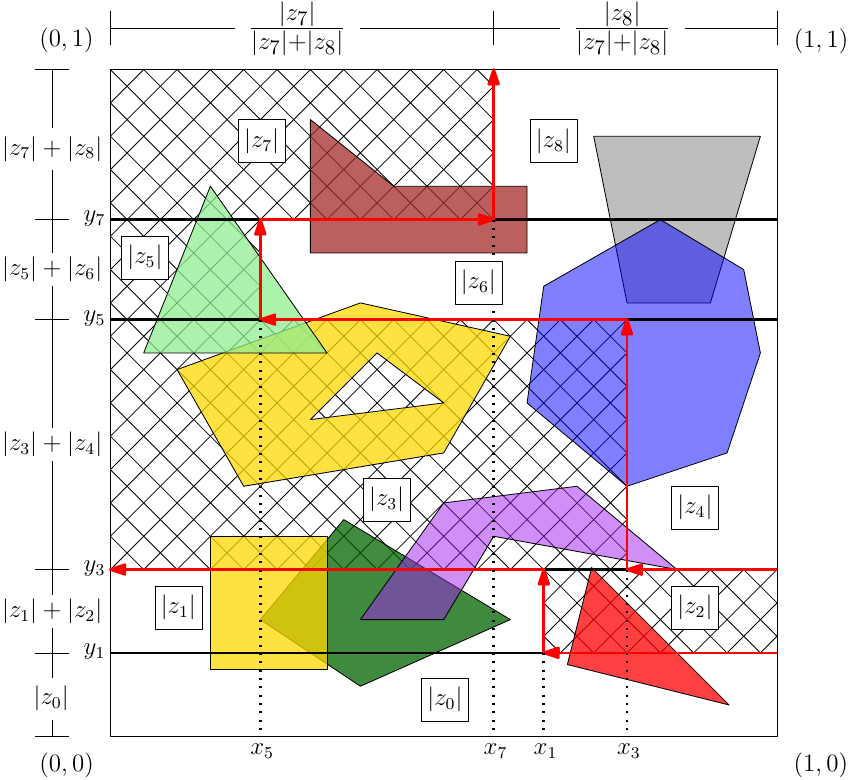}
        \caption{An instance with $n = 8$ polygon-shaped mass distributions and $d=8$. A vector $(z_0, \dots, z_8) \in S^8$ corresponds to horizontal slices and vertical cuts, which define a $y$-monotone \scut-path with $7$ turns. Here, we have $z_0, z_1, z_4, z_6, z_8 > 0$ and $z_2, z_3, z_5, z_7 < 0$. For an example of the figure's notation, notice that in the top strip $[y_7, 1]$, the area to the left of $x_7$ is $|z_7|$ and to the right of $x_7$ is $|z_8|$.}
        \label{fig: BU-incl_1}
	\end{figure}
    
    For the sake of simplicity, we normalize them so that they fit inside $[0,1]^2$ by scaling them down if needed, preserving their relative positions.
    \cref{fig: BU-incl_1} gives an overview of our decoding to \scut-paths, which actually results in $y$-monotone \scut-paths. First, we will consider the case of even $d$, and then explain how the odd $d$ case is proved. We will decode $\vec{z} := (z_0, z_1, \dots, z_d) \in S^d$ into horizontal slices $y_1, y_3, \dots, y_{d-1}$ and vertical cuts $x_1, x_3, \dots, x_{d-1}$. For ease of presentation, let us also define $y_0 := 0$ and $y_{d+1} := 1$. We set $y_1 := | z_0 |$, and $y_j := y_{j-2} + | z_{j-2}| + | z_{j-1} |$, for $j \in \{ 3, 5, \dots, d+1 \}$. It is immediate that $y_{d+1} = \sum_{j=0}^{d} | z_j |$, which equals 1 as it should (since $\vec{z} \in S^d$). So, the bottom \emph{strip} $[y_0, y_1]$ has width $|z_0|$, and the strips $[y_j , y_{j+2}]$ for $j \in \{ 1, 3, \dots, d-1 \}$ have width $| z_j | + | z_{j+1} |$. 
    
    Let us focus on one such strip $[y_j, y_{j+2}]$ to define its respective $x_j$. We have the following cases. (i) If $|z_j| + |z_{j+1}| > 0$ and $z_j \cdot z_{j+1} \geq 0$, then $x_j \in \{ 0, 1 \}$; in particular, if $z_j + z_{j+1} > 0$ then $x_j = 1$, and if $z_j + z_{j+1} < 0$ then $x_j = 0$. (ii) Otherwise, $x_j$ has to satisfy the equation $(|z_j| + |z_{j+1}|) \cdot x_j = | z_j |$; in other words, if $|z_j| + |z_{j+1}| > 0$ and $z_j \cdot z_{j+1} < 0$, then $x_j = \frac{| z_j |}{|z_j| + |z_{j+1}|}$, whereas if $|z_j| + |z_{j+1}| = 0$, then $x_j$ can take any value in $[0,1]$. Notice that in the latter subcase, $z_j = z_{j+1} = 0$, and the width of the strip is 0, so no matter what the value of $x_j$ is, it will not contribute to the \scut-path's structure. 
    
    Finally, we need to label each of the (at most two) parts of each strip. In case (i) above, if $z_j + z_{j+1} > 0$ (resp. $z_j + z_{j+1} < 0$), then the whole strip belongs to \rplus (resp. \rminus), depicted with a non-shaded (resp. shaded) area. In case (ii), if $z_j \geq 0$ and $z_{j+1} \leq 0$ (resp. $z_j \leq 0$ and $z_{j+1} \geq 0$), then the part of the strip to the left of $x_j$ belongs to \rplus (resp. \rminus) and that to the right of $x_j$ belongs to \rminus (resp. \rplus). Similarly, for the strip $[y_0, y_1]$, if $z_0 > 0$ (resp. $z_0 < 0$) then the whole strip belongs to \rplus (resp. \rminus), while if $z_0 = 0$ its width is zero and there is no need to specify where it belongs.

    Now, having the slices $y_j$, the cuts $x_j$, and the labels (\rplus or \rminus) of the slices they define, we can use \cref{alg: cuts_to_paths} to recover the underlying \scut-path.

\begin{algorithm}[h!]
    \caption{Mapping $(z_0, z_1, \dots, z_d)$ to a \scut-path with $d-1$ turns}  
    \label{alg: cuts_to_paths}  
    \begin{algorithmic}[1]
        \REQUIRE{A vector $(z_0, z_1, \dots, z_d) \in S^d$.}
        \ENSURE{A \scut-path with $d-1$ turns.}
        
        \medskip

        \STATE{Let $s_1 := |z_0|$, and $s_j := |z_{j-2}| + |z_{j-1}|$ for $j \in \{ 3, 5, \dots, d+1 \}$ (resp. $j \in \{ 3, 5, \dots, d+2 \}$) when $d$ is even (resp. odd). }
        \STATE{Find the set $T = \{t_1 \dots, t_{r}\} \subseteq \{1, 3, \dots, d+1\}$ (resp. $\{ 1, 3, \dots, d+2 \}$) when $d$ is even (resp. odd), where $t_1 < \dots < t_r$, and for each $\ell \in [r]$ it holds that $s_{t_\ell} > 0$.}
        \IF{$z_0 > 0$}
            \STATE{create an artificial cut $x_0 = 1$ in strip $[y_0, y_1]$, and set $t_0 := 0$}
        \ENDIF
        \IF{$z_0 < 0$}
            \STATE{create an artificial cut $x_0 = 0$ in strip $[y_0, y_1]$, and set $t_0 := 0$}
        \ENDIF
        \IF{$d$ odd \AND $z_{d} > 0$}
            \STATE{create an artificial cut $x_{d} = 1$ in strip $[y_d, y_{d+2}]$}
        \ENDIF
        \IF{$d$ odd \AND $z_{d} < 0$}
            \STATE{create an artificial cut $x_{d} = 0$ in strip $[y_d, y_{d+2}]$}
        \ENDIF
        \STATE{Give an upward direction to all cuts within strips.}
        \STATE{For any given $\ell \in [r]$, let $x_{t_{\ell - 1}} y_{t_{\ell}} x_{t_{\ell}}$ denote the directed horizontal line segment belonging to slice $y_{t_{\ell}}$ that connects the head of cut $x_{t_{\ell - 1}}$ and the tail of cut $x_{t_{\ell}}$. Also, let $\overline{x_{t_{\ell - 1}} y_{t_{\ell}} x_{t_{\ell}}}$ denote its complementary directed line segment with the same start and end points that wraps around the $x$-axis.}
        \STATE{$\ell \gets 1$}
        \WHILE{$\ell \leq r$}
            \IF{$x_{t_{\ell - 1}}, x_{t_{\ell}} \in (0, 1) $ \AND $z_{t_{\ell - 1}} \cdot z_{t_{\ell}} > 0 $}
                \STATE{create $x_{t_{\ell - 1}} y_{t_{\ell}} x_{t_{\ell}}$}
            \ENDIF
            \IF{$x_{t_{\ell - 1}}, x_{t_{\ell}} \in (0, 1) $ \AND $z_{t_{\ell - 1}} \cdot z_{t_{\ell}} < 0 $}
                \STATE{create $\overline{x_{t_{\ell - 1}} y_{t_{\ell}} x_{t_{\ell}}}$}
            \ENDIF

            \IF{$x_{t_{\ell - 1}}, x_{t_{\ell}} \in \{ 0, 1 \} $ \AND $x_{t_{\ell - 1}} \neq x_{t_{\ell}}$}
                \STATE{create $x_{t_{\ell - 1}} y_{t_{\ell}} x_{t_{\ell}}$}
            \ENDIF

            \IF{$x_{t_{\ell - 1}} \in \{ 0, 1 \}$ \AND $x_{t_{\ell}} \in (0, 1) $}
                \IF{$z_{t_{\ell}} > 0 $}
                    \STATE{create $x_{t_{\ell - 1}} y_{t_{\ell}} x_{t_{\ell}}$}
                \ENDIF
                \IF{$z_{t_{\ell}} < 0 $}
                    \STATE{create $\overline{x_{t_{\ell - 1}} y_{t_{\ell}} x_{t_{\ell}}}$}
                \ENDIF
            \ENDIF

            \IF{$x_{t_{\ell - 1}} \in (0, 1)$ \AND $x_{t_{\ell}} \in \{ 0, 1 \} $}
                \IF{$z_{t_{\ell - 1}} > 0 $}
                    \STATE{create $x_{t_{\ell - 1}} y_{t_{\ell}} x_{t_{\ell}}$}
                \ENDIF
                \IF{$z_{t_{\ell - 1}} < 0 $}
                    \STATE{create $\overline{x_{t_{\ell - 1}} y_{t_{\ell}} x_{t_{\ell}}}$}
                \ENDIF
            \ENDIF
            
            \STATE{$\ell \gets \ell+1$}
        \ENDWHILE
        \STATE{Remove artificial cuts from strips $[y_0, y_1]$ and $[y_d, y_{d+2}]$ (odd $d$ case), if any.}
    \end{algorithmic}
\end{algorithm}

    What remains is the definition of the required Borsuk-Ulam function. For any given point $\vec{z} = (z_0, z_1, \dots, z_d) \in S^{d}$, the Borsuk-Ulam function is defined to be the total ``$+$'' (positive) measure on $[0,1]^2$ induced by $\vec{z}$, and we denote $f_{i}(\vec{z}) = \mu_{i}(\rplus;\vec{z})$ for $i \in [n]$. The total positive measure $\mu_{i}(\rplus;\vec{z})$ is a continuous function of its variables: by our definition of slices, cuts, and labelling, the area of \rplus is the sum of individual areas $| z_j |$ for those $j \in \{ 0, 1, 2, \dots, d \}$ for which $z_j > 0$ (see non-shaded area of \cref{fig: BU-incl_1}); the slices $y_j$ and the cuts $x_j$ are continuous functions of the $z_j$'s;\footnote{It is true that, e.g., for some fixed $z_{j+1} > 0$ and some moving $z_j \to 0^-$ we have $x_j \to 0^+$, while when $z_j = 0$, immediately $x_j = 1$. However, this is considered a continuous behaviour of $x_j$ since it is allowed to wrap around in the horizontal dimension. Also, notice that this has no effect on the sign of any other individual area, creating no discontinuities. Moreover, the only case where the cut of a strip can arbitrarily take values in $(0,1)$ is when for its components $z_j , z_{j+1}$ we have $z_j = z_{j+1} = 0$, in which case the strip's width is also 0, and so it does not contribute to $\mu(\rplus;\vec{z})$.} in each such individual area, the contained measure $j \in \{ 0, 1, \dots, d \}$ changes continuously with the boundaries of the area; also, one can easily check that in order for an individual area corresponding to $z_j$ to change sign, $z_j$ will have to pass from 0, and the magnitude $|z_j|$ of the area is a continuous function of $z_j$.

    When $d$ is odd, the decoding of $\vec{z} = (z_0 , z_1, \dots, z_d)$ into a \scut-path is similar. The horizontal slices are $y_1, y_3, \dots, y_d$, and we set $y_0 := 0$ and $y_{d+2} := 1$. These define strips similarly to the even $d$ case. The vertical cuts are $x_1, x_3, \dots, x_{d-2}$, meaning that the bottom strip $[y_0, y_1]$ \emph{and} the top strip $[y_d, y_{d+2}]$ are not vertically cut. 
    
    Also, it is easy to see that one could consider the path to be again $y$-monotone but in the opposite direction, meaning that there is no line segment pointing upwards.
	
    Given $n$ measures, if $d = n$ the Borsuk-Ulam theorem applies on $f$, ensuring that there exist two antipodal points $\vec{z}^*, -\vec{z}^* \in S^{n}$ such that $f(\vec{z}^*) = f(-\vec{z}^*)$. Notice that for any $i \in [n]$, $f_{i}(-\vec{z}) = \mu_{i}(\rminus;\vec{z})$, since by flipping the signs of the variables of $\vec{z}$, we consider the ``$-$'' (negative) measure of $[0,1]^2$ induced by (the \scut-path of) $\vec{z}$. Therefore, when $f(\vec{z}^*) = f(-\vec{z}^*)$ we will have $\mu_{i}(\rplus;\vec{z}^*) = \mu_{i}(\rminus;\vec{z}^*)$ for every $i \in [n]$, that is, in each of the $n$ measures, the positive total measure equals the negative one. Therefore, the \scut-path corresponding to $\vec{z}^*$ (see \cref{alg: cuts_to_paths}) is a solution to \scut-\pizza. The total number of turns of the directed path is $d-1 = n-1$.
    \end{proof}

    \begin{remark}
        Note that a symmetric proof exists, where the slices are vertical instead of horizontal, and the cuts within the strips are horizontal instead of vertical. The analysis is similar to the one we give here, and it guarantees the existence of a \scut-path which is allowed to wrap around in the vertical dimension, it bisects all $n$ measures and is $x$-monotone with either no line segment heading left or no line segment heading right.
    \end{remark}
	
	\myparagraph{\ppa/ membership.} 
	The following theorem shows \ppa/ membership of \eps-\scut-\pizza via a reduction to the \epsborsuk problem which is in \ppa/. The latter problem was introduced and shown to be in \ppa/ by \cite{Papadimitriou94-TFNP-subclasses} with its definition involving essentially a polynomial-time algorithm for the computation of the Borsuk-Ulam function. In \cite{DFMS21}, the definition of the problem uses an algebraic circuit as the representation of that function. The \ppa/ membership is shown via a reduction to \tucker (see \cite{Papadimitriou94-TFNP-subclasses}), and for both versions of the \epsborsuk problem the reduction goes through. Here we state the most inclusive version of the problem.
	
	\begin{mdframed}[backgroundcolor=white!90!gray,
		leftmargin=\dimexpr\leftmargin-20pt\relax,
		innerleftmargin=4pt,
		innertopmargin=4pt,
		skipabove=5pt,skipbelow=5pt]
		\begin{definition}[\epsborsuk]\label{def: eps-Borsuk-Ulam_prob} \
			\begin{itemize}
				\item \textbf{Input:} 
				$\eps > 0$, and a continuous function $f: S^d \to \reals^d$ whose Lipschitz constant is claimed to be $\lambda$. The function can be presented as either an algebraic circuit or a polynomial-time algorithm.
				\item \textbf{Task:} Find one of the following. 
				\begin{enumerate}
                        \item[(a)] Two points $x, y \in S^d$ such that $\|f(x) - f(y)\|_\infty > \lambda
					\cdot \|x - y\|_\infty$. 
					\item[(b)] A point $x \in S^d$ such that $\|f(x) - f(-x)\|_\infty \le \eps$.
				\end{enumerate}
			\end{itemize}
		\end{definition}
	\end{mdframed} \vspace{5pt}
	If the first task is accomplished, then we have found witnesses $x, y \in S^d$ that function $f$ is not $\lambda$-Lipschitz continuous in the $L_{\infty}$-norm as required.\footnote{The $\lambda$-Lipschitzness of $f$ is necessary for the correctness of the reduction from \epsborsuk to \tucker in \cite{DFMS21}, where the latter problem is known to be in \ppa/. In particular, the reduction triangulates the $S^d$ sphere such that the triangulation's vertices have distance at most $O(\eps/\lambda)$, therefore, if $f$ is not Lipschitz continuous (i.e., $\lambda$ is unbounded), the solutions of \tucker are not guaranteed to correspond to \epsborsuk solutions.} But if the second task is accomplished then we have an approximate solution to the Borsuk-Ulam problem.
	
	\begin{theorem} \label{thm: eps-sc-pizza-in_ppa}
		\eps-\scut-\pizza for weighted polygons with holes is in \ppa/.
	\end{theorem}

    \begin{proof}
    We will turn our existence proof of \cref{thm: sc-pizza_exists} into a polynomial time reduction from \eps-\scut-\pizza to \epsborsuk in which the Borsuk-Ulam function is computable via a polynomial-time algorithm. Given the \eps-\scut-\pizza instance with $n$ mass distributions, we just have to construct the Borsuk-Ulam function $f: S^n \to \reals^n$ using the procedure described in \cref{app: exact-computation-measures}, for $d = n$. The entire procedure involving the preprocessing part and the construction of $f$ is a polynomial-time algorithm. 
    
    Furthermore, the function captures the \rplus part of each involved colour $i \in [n]$ by creating a \scut-path as described in the proof of \cref{thm: sc-pizza_exists}, where we showed that $f$ is continuous. Also, it is easy to verify from the final step of the construction in \cref{app: BU-func_comp} that $f$ is piece-wise polynomial with respect to $\vec{z}$, and therefore it is $\lambda$-Lipschitz continuous for $\lambda = \max_{j=0}^{n}\left\{ \sup_{\vec{z}} \left\|\frac{\partial f(\vec{z})}{\partial z_j}\right\|_{\infty}\right\}$ (and note that points where $f_{i}(\vec{z})$, $i \in [n]$ is non-differentiable do not matter for Lipschitzness). By the construction of $f$ (\cref{app: BU-func_comp}), one can see that $\lambda$ is constant: the polynomial pieces of $f$ are of degree at most 2, and the partial derivative of each $f_{i}$ is determined by the rational points (given in the input) that define the polygons. 
    
    So far we have showed how to formulate any given instance of \eps-\scut-\pizza as an \epsborsuk instance in polynomial time. What remains is to show how to turn a solution $\vec{z}^*$ of the latter to a solution of the former again in polynomial time. As we showed in the proof of \cref{thm: sc-pizza_exists}, any $\vec{z} \in S^n$ can be efficiently translated into a \scut-path using \cref{alg: cuts_to_paths}. The \scut-path corresponding to $\vec{z}^*$ (for which we have $\| f(\vec{z}^*) - f(-\vec{z}^*) \|_{\infty} \leq \eps$) is the solution to \eps-\scut-\pizza. To see this, notice that from the aforementioned proof we have $f_{i}(\vec{z}^*) = \mu_{i}(\rplus; \vec{z}^*)$ and $f_{i}(-\vec{z}^*) = \mu_{i}(\rminus; \vec{z}^*)$, therefore, $\| \mu_{i}(\rplus; \vec{z}^*) - \mu_{i}(\rminus; \vec{z}^*) \|_{\infty} \leq \eps$. Finally, since $\vec{z}^* \in S^n$, the \scut-path has at most $n-1$ turns as required by an \eps-\scut-\pizza solution.
    \end{proof}

    \myparagraph{\np/ membership.} 
    Here we show that for any $k, n \in \mathbb{N}$, deciding whether there exists a solution for \eps-\scut-\pizza with $n$ measures and at most $k$ turns is in \np/. Notice that, for any such instance, we can verify a candidate solution in polynomial time. In particular, suppose we are given a \scut-path with at most $k$ turns that splits $[0,1]^2$ into \rplus and \rminus regions. Let the path be represented by the starting and ending points of its line segments, and the regions be defined by labels to the left and right of each vertical segment. Each measure $i \in [n]$ consists of a set of polygons with holes which can be preprocessed in polynomial time as described in \cref{app: axis-aligned-tr} to result in only axis-aligned right-angled triangles. Then, using the \scut-path, the measures $\mu_i(\rplus)$ and $\mu_i(\rminus)$ (in a similar way) can be computed in polynomial time as described in \cref{app: BU-func_comp}, where now $d = k+1$. Finally, for the given \eps, we can check whether $\left| \mu_i(\rplus) - \mu_i(\rminus) \right| \leq \eps$ is true for all $i \in [n]$. Therefore, we get the following.
    \begin{theorem} \label{thm: NP-incl}
        Deciding whether there exists a \scut-path with at most $k \in \naturals$ turns that is a solution of \eps-\scut-\pizza with $n \in \naturals$ mass distributions is in \np/.
    \end{theorem}

        \subsection{Membership results for discrete \scut-\pizza} \label{sec: discrete-pizza_containment}

        It has already been shown in \cite{S21} that \eps-\dspizza is in \ppa/ even for $\eps = 0$. We complete the picture regarding inclusion of discrete pizza sharing problems, by showing that \eps-\dpizza is also in \ppa/ for $\eps = 0$, and therefore, for every $\eps \in [0,1]$.
        We will reduce \eps-\dpizza to $\eps'$-\scut-\pizza for $\eps = 0$ and $\eps' = 1/2N$, where $N$ is the input size. Finally, as one would expect from the discrete version, its decision variant is in \np/ since its candidate solutions are verifiable in polynomial time.

        \myparagraph{\ppa/ membership.}
        Consider an instance of \eps-\dpizza with point sets $P_1 , \dots, P_n$, and denote $P := P_1 \cup \dots \cup P_n$. We intend to turn each point into a mass of non-zero area. To do so, we need to first scale down the landscape of the points, to create some excess free space as a ``frame'' around them. It suffices to scale down by an order of 3 and center it in the middle of $[0,1]^2$, that is, to map each point $(x,y)$ to $\left(\ \frac{1}{3} + \frac{x}{3}, \frac{1}{3} + \frac{y}{3} \right)$. Now all our points are in $[1/3, 2/3]^2$. For convenience, for each $i \in [n]$, we will be still denoting by $P_i$ the new set of points after scaling and centering.
        
        Now, we check whether for any pair $i \neq j$ we have $P_{i,j} := P_i \cap P_j \neq \emptyset$, which means that two points of two point sets have identical positions. Consider all $P_{i,j} \neq \emptyset$ and let their union be $P'$. Now consider all points that do not belong in $P'$, that is $R := P \setminus P'$. We want to find the minimum positive difference in the $x$- and $y$-coordinates between any pair of points in $R$. Let a point of $R$ be denoted $p_t = (x_t , y_t)$, and let us denote 
        \begin{align*}
            x_{\min} := \min_{\substack{p_a, p_b \in R\\ x_a \neq x_b}} |x_a - x_b| , \qquad \text{and by} \qquad y_{\min} := \min_{\substack{p_a, p_b \in R\\ y_a \neq y_b}} |y_a - y_b|,
        \end{align*}
        and finally, $d := \min \{ x_{\min}, y_{\min} \}$.

        We now turn each point of $P_i$ into an axis-aligned square of size $\frac{d}{3} \times \frac{d}{3}$, with its bottom-left corner having the point's coordinates. Notice that the total area of the squares is $|P_i| \cdot \frac{d^2}{9}$, therefore, by setting the weight of each square to $\frac{9}{|P_i| d^2}$ we have the full description of a mass distribution $\mu_i$. Notice that, since $d \leq 1/3$, all of the mass distributions are in $[1/3 - 1/3 \cdot 3, 2/3 + 1/3 \cdot 3]^2 = [2/9, 7/9]^2$.
        \begin{lemma}\label{lem: discr-to-cont-sc-pizza}
            Any \scut-path that is a solution to the resulting \eps-\scut-\pizza instance for $\eps = 1/2N$, can be turned into a solution of $\eps'$-\dpizza for $\eps' = 0$ in polynomial time.
        \end{lemma}

        \begin{proof}
            If a horizontal (resp. vertical) line segment of the \scut-path solution intersects two squares (of any two mass distributions), this means that their corresponding points in \dpizza had the same $y$- (resp. $x$-) coordinate. To see this, without loss of generality, suppose that the two squares are intersected by the same horizontal line segment, and that their corresponding points do not have the same $y$-coordinate. Then, the distance between their bottom-left corners is positive but no greater than $d/3$, which implies that $d \leq d/3$ (by definition of $d$), a contradiction. A symmetric argument holds for two squares that are intersected by the same vertical line segment. In both the above cases, we return as a solution to $0$-\dpizza the two corresponding points of the squares, which are of the kind of ``Output (a)'' in \cref{def:discrete-sc-pizza}.

            Let $P_{\max} := \max_{i \in [n]} |P_i|$. 
            Suppose $|P_i|$ is odd for every $i \in [n]$. Then, since we are asking for an $1/2N$-\scut-\pizza solution, its \scut-path cannot be non-intersecting with any of the squares, otherwise $|\mu_i(\rplus) - \mu_i(\rplus)| \geq 1/|P_i| \geq 1/P_{\max} > 1/2P_{\max} \geq 1/2N$, a contradiction. Therefore, at least one square of $P_i$ is intersected, and this holds for every $i \in [n]$. 
            If no line segment of the \scut-path intersects two squares, we conclude that the \scut-path with $n-1$ turns and $n$ line segments will intersect at most $n$ squares. But since, as discussed above, at least one square of each mass distribution has to be intersected by \scut, we get that \scut intersects exactly one square of each mass distribution. Each side, \rplus and \rminus of the \scut-path includes at least $\floor{|P_i| / 2}$ entire squares for every $i \in [n]$, and therefore, their bottom-left corners, i.e., the corresponding points of \dpizza. This is a solution to the $0$-\dpizza instance. 

            Now suppose $|P_i|$ is even for some $i \in [n]$. We can remove an arbitrary square from all mass distributions that come from point sets with even cardinality, and perform the aforementioned reduction to $1/2N$-\scut-\pizza. Then, let us call ``$i$-th segment'' the one that intersects one square of $P_i$, called $i$-th square, and let it be a vertical segment, without loss of generality. By placing back the removed square, its bottom-left corner will be: (i) either on opposite sides with that of the $i$-th square, (ii) or on the same side (notice that due to the allowed discrepancy $1/2N < 1/2P_{\max}$, the $i$-th segment cannot fall on the bottom-left corner of the $i$-th square). Then, in case (i) each side contains exactly $|P_i|/2$ bottom-left corners of squares (i.e., points). By scaling up the positions of \scut-path's segments (recall that we have scaled down), this is a solution to the $0$-\dpizza instance. In case (ii), suppose without loss of generality that both the inserted square and the bottom-left corner of the $i$-th square are on the left side of the $i$-th segment. We modify the \scut-path by shifting the $i$-th segment to the left such that it is now located $d/3$ to the left of $i$-th square's bottom-left corner. Notice that this position is to the right of the inserted square since there is at least $2d/3$ distance between the two squares. We do the same for every $i \in [n]$ has even number of points/squares. Then, each side of the \scut-path for every $i \in [n]$ has exactly $|P_i|/2$ bottom-left corners of squares which represent the initial points. After scaling up the positions of the \scut-path's segments, this is a solution to $0$-\dpizza.

            Finally, it is clear that the aforementioned operations can be performed in $\poly(N)$ time.
        \end{proof}

        By the \ppa/ membership of $1/2N$-\scut-\pizza (see \cref{thm: eps-sc-pizza-in_ppa}), we get the following.
        \begin{theorem}\label{thm: discr-sc-pizza-ppa-contain}
            \eps-\dpizza is in \ppa/ for any $\eps \in [0,1]$.
        \end{theorem}

        \myparagraph{\np/ membership.}
        It is also easy to see that, by checking whether each of the points of each $P_i$ is in \rplus or \rminus as defined by a candidate \scut-path solution, we can decide in polynomial time if indeed it is a solution or not to $0$-\dpizza (since the points are polynomially many in $N$, by definition).
        \begin{theorem}\label{thm: discr-sc-pizza-np-contain}
            Deciding whether there exists a \scut-path solution to \eps-\dpizza is in \np/ for any $\eps \in [0,1]$.
        \end{theorem}

	\subsection{\etr/ membership for the decision variant of exact \scut-\pizza} \label{sec: exact_inclusion}

    Here we show that deciding whether there exists an exact \scut-\pizza solution (\scut-path) with at most $k$ turns and $n$ measures is in \etr/, for any $k, n \in \naturals$. To do so, we turn our decision problem into an ETR formula in polynomial time. As discussed in \cref{sec:preliminaries}, an ETR expression has the form: $\exists \vec{P} \in \reals^{m} \cdot \Phi$, where $\Phi$ is a Boolean formula using connectives $\{\land, \lor, \lnot\}$ over polynomials with domain $\reals^{m}$ for some $m \in \naturals$ compared with the operators $\{ <, \leq, =, \geq, > \}$. The ETR problem is to decide whether there is a truth assignment $\vec{P} \in \reals^{m}$ that satisfies $\Phi$. 
    
    We will use the proof of \cref{thm: sc-pizza_exists} and the explicit construction of the Borsuk-Ulam function from \cref{app: exact-computation-measures}. Recall that the aforementioned sections provide a polynomial time algorithm that takes the problem's input, i.e., $n$ sets of polygons with holes, and computes a Borsuk-Ulam function, $f : S^d \to \reals^n$, for some given $d \in \naturals^*$. This is done by first mapping any point $\vec{z} \in S^d$ to a \scut-path with at most $d-1$ turns, and then computing $f_i$, $i \in [n]$, which captures the $i$-th measure's intersection with the \rplus region (where the \rplus, \rminus regions have been determined by the \scut-path). In \cref{app: BU-func_comp} we showed how to explicitly construct $f$ in polynomial time and showed that it is piece-wise polynomial.

    Now notice that for every $i \in [n]$, $f_{i}(\vec{z}) = \mu_i(\rplus ; \vec{z})$, therefore, $f_{i}(\vec{z}) = f_{i}(-\vec{z})$ is equivalent to $\mu_i(\rplus ; \vec{z}) = \mu_i(\rplus ; -\vec{z}) = \mu_i(\rminus ; \vec{z})$, where the last equality comes by the definition of the \scut-path decoded from $\vec{z}$. In other words, $\vec{z} \in S^d$ is a solution to $f(\vec{z}) = f(-\vec{z})$ if and only if its corresponding \scut-path is a solution to $\mu(\rplus ; \vec{z}) = \mu(\rminus ; \vec{z})$. Let us set $d = k+1$. What remains is to turn the decision problem of whether such a $\vec{z}$ exists into an ETR formula.

    We start by replacing the domain $\vec{z} \in S^{k+1}$ with $\vec{Z} \in \reals^{k+2}$ and adding in the ETR formula the constraint $\sum_{j=0}^{k+1} |Z_j| = 1$. Then, we turn all the aforementioned polynomial time computations (that result to $f$) into ETR form. We can use a standardized method to do so in a generic manner by efficiently expressing in ETR form any computation belonging to \np/. In particular, it is clear that any such computation can be turned in polynomial time into a Boolean satisfiability (SAT) formula, and sequentially transform it into a CNF formula. To turn this formula into an ETR expression is easily done in the following way (e.g., see \cite{Basu2006}). Consider the CNF formula $B$ over $m \in \naturals$ Boolean variables $\{ x_1, \dots, x_m \} $. This can be turned in polynomial time into an equisatisfiable ETR formula: $\exists \vec{X} \in \reals^{m} \cdot \bigwedge_{i=1}^{m} ((X_i = 0) \lor (X_i = 1)) \land B'$, where $B'$ is constructed in the following way. For each $i \in [m]$, let $y_i \in \{ x_i, \lnot x_i \}$ be a literal in $B$. We turn all disjunctions $y_j \lor y_k \lor \dots \lor y_\ell$ of $B$ into the inequality $Y_j + Y_k + \dots + Y_\ell > 0$ in $B'$, where for each $i \in [m]$, $Y_i = X_i$ if $y_i = x_i$ and $Y_i = 1-X_i$ if $y_i = \lnot x_i$. Finally, the auxiliary variables $F_i$, $F'_i$, $i \in [n]$, contain the values of $f_{i}(\vec{z})$'s and $f_{i}(-\vec{z})$'s computed from $B'$. 
    
    The above induce the following ETR expression, which is true if and only if there is an exact solution (\scut-path) with at most $k$ turns for the \scut-\pizza problem with $n$ measures.
	\begin{align*}
		\exists (\vec{Z}; \vec{X}; \vec{F}) \in \reals^{k+2 + m + 2n} \cdot \left( \sum_{j=0}^{k+1} |Z_j| = 1 \right) \land \bigwedge_{i=1}^{m} \left( (X_i = 0) \lor (X_i = 1) \right) \land B' \land \left( \bigwedge_{i=1}^{n} F_{i} = F'_{i} \right).
	\end{align*}

    This gives us the following theorem.
	
    \begin{theorem} \label{thm: ETR-incl}
		Deciding whether there exists a \scut-path with at most $k \in \naturals$ turns that is an exact solution for \scut-\pizza with $n \in \naturals$ mass distributions is in \etr/.
    \end{theorem}

	\section{Conclusions}

        For \eps-\spizza we have shown that finding a solution is \ppa/-complete for any $\eps \in \left[ 1/N^c , 1/5 \right)$, where $N$ is the input size and $c > 0$ is a constant. This result holds for both its continuous (even when the input contains only axis-aligned squares) and its discrete version. We have also shown that the same result holds for \eps-\scut-\pizza, where the \ppa/ membership holds even for inverse exponential \eps. One open question that remains is ``Can we prove \ppa/ membership of \eps-\spizza for inverse exponential \eps?''. For the decision variant of both these problems, we show that there exists a small constant \eps such that they are \np/-complete. For both these problems and their search/decision variants, a big open question is ``What is the largest constant \eps for which the problem remains \ppa/-hard and \np/-hard, respectively?''. The most interesting open question is ``Are there any algorithms that guarantee a solution in polynomial time for some constant $\eps \in [1/5, 1)$, even when slightly more lines (resp. turns in a \scut-path) are allowed?''.

        We have also shown that exact \scut-\pizza is \fixp/-hard. The interesting question that needs to be settled are ``Is the problem in the class \bu/ (defined in \cite{DFMS21}) similarly to exact \ch?'', and ``For what complexity class is the problem complete?''. Schnider in \cite{S21} showed that exact \spizza is \fixp/-hard for a more general type of input than ours. So, some natural questions are ``When the input consists of weighted polygons with holes, is the problem \fixp/-hard and/or inside \bu/?'', and ``Is it complete for any of the two classes?''. For a strong approximation version of \ch, \cite{BHH21} showed that the problem is $\bu/_{a}$-complete. We conjecture that the same holds for the two pizza sharing problems studied here.

        In the \scut-\pizza problem, a path is allowed to wrap around on either the $x$-axis or the $y$-axis, suggesting a cylindrical shape of the underlying space. It would be intriguing to study the case of a torus or even that of a plane.
        Another couple of questions that remain open are: ``What is the complexity of \eps-\spizza and \eps-\scut-\pizza when every mass distribution consists of a constant number of non-overlapping rectangles?'', and finally, ``What is the complexity of the pizza sharing problems when we ask to fairly split the plane into more than two equal parts?''.

	\appendix

	\section{Proof of Lemma \ref{lem: cont-to-discr}}
	\label{app: cont-to-discr}
	
	\myparagraph{Pixelation.}
	We will start with the task of pixelating the mass distributions. Consider the input of an \eps-\scut-\pizza or an \eps-\spizza instance, that is, $q \in \{n, 2n\}$ mass distributions, respectively, on $[0,1]^2$ consisting of weighted polygons with holes (see \cref{sec:preliminaries} for details on the input representation). Let the instance's input size be $N \geq 2n$ (by definition). As a first step, we perform a \emph{pixelation} procedure: each polygon will be turned into a union of smaller squares that will have approximately the same total area as the polygon.
	
	As shown in \cref{app: axis-aligned-tr}, it is possible to decompose a polygon into a union of disjoint non-obtuse triangles (\cref{prop: non-obtuse_triangle_decomp}). Each of those triangles' area is rational since it can be computed by adding and subtracting the areas of five right-angled triangles with rational coordinates, which additionally, are axis-aligned (\cref{prop: ax-al_ri-tr_decomp}). Furthermore, the cardinality of the non-obtuse triangles is a polynomial function in the input size of the polygon's description, that is, the coordinates of the points that define its corners and the value that defines its weight. Therefore, the exact area of any mass distribution can be computed in polynomial time.
	
	As we have discussed earlier, in order for our approximation parameter $\eps \in [0,1]$ to make sense, we consider normalized mass distributions, meaning that $\mu_i \left( [0,1]^2 \right) = 1$ for all $i \in [q]$. Notice that it can be the case that some mass distributions have total area constant (in which case their weight is constant), while some others might have total area exponentially small (and therefore exponentially large weight). Therefore, in our analysis, we make sure that the ``resolution'' we provide to any polygon $F$ is \emph{relative} to its actual area $\area(F)$ rather than its measure $\mu_i(F)$.

    \begin{figure}[htbp]
        \begin{subfigure}[b]{0.47\textwidth}
            \centering
            \includegraphics[width=0.7\linewidth]{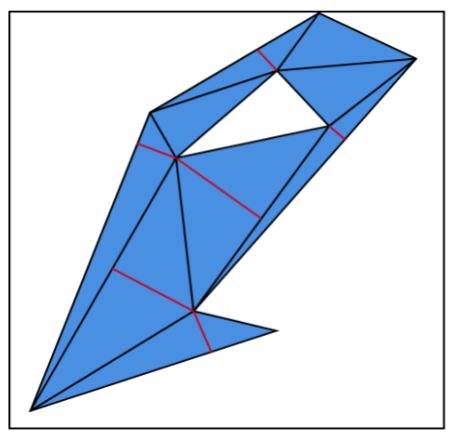}
            \caption{The triangulation with only non-obtuse triangles. After the standard triangulation, extra line segments (in red colour) are added to ensure non-obtuseness.} \label{fig:combined-a}
        \end{subfigure}%
        \hspace*{\fill}   
        \begin{subfigure}[b]{0.47\textwidth}
            \centering
            \includegraphics[width=0.68\linewidth]{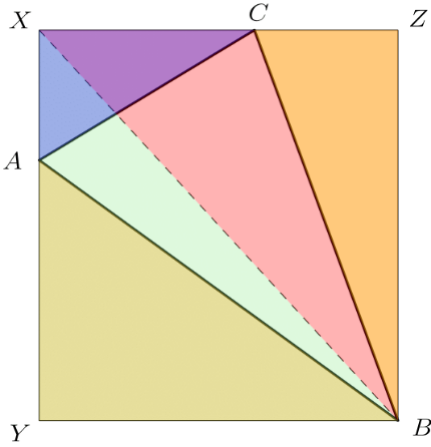}
            \caption{A non-obtuse triangle $\trngl{ABC}$ and its decomposition into axis-aligned right-angled triangles: $\area(\trngl{ABC}) = \area(\trngl{XYB}) + \area(\trngl{XBZ}) - \area(\trngl{AYB}) - \area(\trngl{XAC}) - \area(\trngl{CBZ})$.} \label{fig:combined-b}
        \end{subfigure}%
        \caption{Decomposing a polygon into a set of axis-aligned right-angled triangles.}
        \label{fig:combined}
    \end{figure}
	
	Consider some polygon $F$ of mass distribution $\mu_i$ for $i \in [q]$ and let it be triangulated into non-obtuse triangles, all with strictly positive area. We will focus on  one of $F$'s non-obtuse triangles, $\trngl{ABC}$ (see \cref{fig:combined-b}) with area $S := \area \left( \trngl{ABC} \right) > 0$ and perimeter $T > 0$. Notice that since $\trngl{ABC}$ is in $[0,1]^2$, we have $S \leq 1$ and $T \leq 3 \cdot \sqrt{2} < 5$.
	Suppose that among all triangles of all $\mu_i$'s, the minimum area triangle has area $\alpha$. The first step is to pixelate $\trngl{ABC}$. Let our \emph{pixels} be thought of as squares of size $t \times t$ for $t := \frac{1}{\ceil{15N^{1+c}/\alpha}}$, where $c > 0$ is any fixed constant. In other words, consider an axis-aligned square grid on $[0,1]^2$ with edge length $t$, where the closed region defined by four edges is called a pixel (see \cref{fig:pixels-1} for a depiction). 
	We create a pixel for $\trngl{ABC}$ if and only if the pixel's intersection with $\trngl{ABC}$ has strictly positive area. Then, the pixelated version of the triangle, denoted $\trngl{ABC}_p$, has area $S + S'$, where $S'$ is the excess area induced by the pixels intersected by the three sides of the triangle. By definition, $S'$ is at least $0$, and at most the area of pixels that intersect the three sides of $\trngl{ABC}$. Therefore, by denoting the number of such pixels for each side by $n_{AB}, n_{BC}, n_{CA}$ and referring to \cref{fig:combined-b}, we have $n_{AB} \leq \ceil{\frac{AY}{t}} + 1 + \ceil{\frac{YB}{t}} + 1 \leq \frac{AY}{t} + \frac{YB}{t} + 4$, and similarly for $n_{BC}, n_{CA}$. Now also notice that $t \leq \frac{\alpha}{15N^{1+c}} \leq \frac{S}{15} < \frac{S}{3 T}$, where the second inequality comes from the fact that $S \geq a$ by definition, and the third one is due to $T < 5$ as argued above. Then, we also have to use the known formula that connects $S$ and $T$, namely,
	\begin{align*}
		S = \sqrt{\frac{T}{2} \left(\frac{T}{2}-AB \right) \left(\frac{T}{2}-BC \right) \left(\frac{T}{2}-CA \right)} < \sqrt{\left(\frac{T}{2}\right)^4} = \frac{T^2}{4},
	\end{align*}
	which implies $\frac{S}{T} < \frac{T}{4}$, and therefore $t < \frac{T}{12}$.
 
    Putting everything together, we get 
	\begin{align}\label{eq: extra-pix}
		S' &\leq t^2 \cdot \left( n_{AB} + n_{BC} + n_{CA} \right) \nonumber  \\ 
        &\leq t^2 \cdot \left( \frac{AY}{t} + \frac{YB}{t} + \frac{BZ}{t} + \frac{ZC}{t} + \frac{CX}{t} + \frac{XA}{t} + 12 \right) \nonumber  \\
		&\leq t \cdot \left( (AY + YB) + (BZ + ZC) + (CX + XA) + 12t \right) \nonumber  \\
		&\leq  t \cdot \left(2 \cdot AB + 2 \cdot BC + 2 \cdot CA + 12t \right) \qquad \text{(since $AB, BC, CA$ are hypotenuses)} \nonumber   \\
		&=t \cdot \left(2 \cdot T + 12t \right) \nonumber \\
		&< 3 \cdot T \cdot t   \qquad \text{(since $t < \frac{T}{12}$)} \nonumber   \\
		&< \frac{\alpha}{N^{1+c}} \qquad \text{(since $T < 5$)} .
	\end{align}

	We want to bound the proportion of excess area due to the pixelation compared to the triangle's actual area, that is, $S'/S$. We have 
	\begin{align*}
		\frac{S'}{S} < \frac{\alpha}{S \cdot N^{1+c}} , 
	\end{align*}
	which, together with he fact that $S \geq \alpha$ (by definition of $\alpha$), gives the following.
	\begin{claim}\label{clm: pix-inflation}
		The pixelation of $\trngl{ABC}$ results in $\trngl{ABC}_p$, where $\area \left( \trngl{ABC}_p \right) < \left( 1+\frac{\alpha}{S \cdot N^{1+c}} \right) \cdot \area \left(\trngl{ABC} \right)$. In particular, $\area \left( \trngl{ABC}_p \right) < \left( 1+\frac{1}{N^{1+c}} \right) \cdot S $.
        
	\end{claim}
	
	Consider a straight line $\ell$ that cuts $\trngl{ABC}_p$, splitting it into two shapes $L_p$ and $R_p$ with areas $\area(L_p)$ and $\area(R_p)$, respectively. Let the corresponding two shapes that $\ell$ creates when intersecting $\trngl{ABC}$ be $L$ and $R$, with areas $\area(L)$ and $\area(R)$, respectively. Also, let $M$ be the set of pixels of $\trngl{ABC}_p$ that are intersected by $\ell$, and $M_L, M_R$ be a partition of $M$. When clear from the context, we will slightly abuse the notation by denoting $D$ the set of points in the union of squares on $[0,1]^2$ corresponding to pixel set $D$.
	
	\begin{claim}\label{clm: pix-intersection}
		The total area of $M$ is at most $\alpha/5N^{1+c}$.
	\end{claim}
	
	\begin{proof}    
		We start from the easy observation that the length of $\ell \cap \trngl{ABC}_p$ is at most $\sqrt{2} < 2$ since we are in $[0,1]^2$. Therefore, the number of pixels that $\ell$ intersects is at most $\frac{2}{t}$ resulting to their area being at most $t^2 \cdot \frac{2}{t} = 2t < \frac{\alpha}{5N^{1+c}}$.
	\end{proof}
	
	\begin{claim}\label{clm: pix-gap}
		For any disjoint $M_L, M_R$ with $M_L \cup M_R = M$, we have $| (\area(L) - \area(R)) - (\area(L_p \cup M_L) - \area(R_p \cup M_R) )| \leq  2\alpha/N^{1+c}$.
	\end{claim}
	
	\begin{proof}
		It suffices to show that $0 \leq \area(L_p \cup M_L) - \area (L) \leq 2\alpha/N^{1+c}$. The first inequality is easy to see, since $L \subseteq L_p \cup M_L$, which implies $\area (L) \leq \area (L_p \cup M_L)$. For the second inequality, we have
		\begin{align*}
			\area (L) &> \frac{\area(L_p)}{1+\alpha/SN^{1+c}} \qquad \text{(by \cref{clm: pix-inflation})}  \\
			&= \area(L_p) \cdot \left(1 - \frac{\alpha}{SN^{1+c} + \alpha} \right) \\
			&\geq \left( \area(L_p \cup M_L) - \area(M) \right) \cdot \left(1 - \frac{\alpha}{SN^{1+c}} \right) \\
			&\geq \area(L_p \cup M_L) - \frac{\alpha}{5N^{1+c}} - \frac{\alpha}{SN^{1+c}} \cdot \left( \area(L_p \cup M_L) - \frac{\alpha}{5N^{1+c}} \right) \qquad \text{(by \cref{clm: pix-intersection})} \\
			&\geq \area(L_p \cup M_L) - \frac{\alpha}{5N^{1+c}} - \frac{\alpha}{SN^{1+c}} \cdot \area \left( \trngl{ABC}_p \right)  \\
                &> \area(L_p \cup M_L) - \frac{\alpha}{5N^{1+c}} - \frac{\alpha}{SN^{1+c}} \cdot S \left(1 + \frac{1}{N^{1+c}}\right) \qquad \text{(by \cref{clm: pix-inflation})} \\
                &\geq \area(L_p \cup M_L) - \frac{\alpha}{5N^{1+c}} - \frac{3}{2} \frac{\alpha}{N^{1+c}}  \qquad \text{(since $N \geq 2$)} \\
			&\geq \area(L_p \cup M_L) - \frac{2\alpha}{N^{1+c}}.
		\end{align*}

		Similarly, $0 \leq \area(R_p \cup M_R) - \area (R) \leq 2\alpha/N^{1+c}$, or equivalently, $-2\alpha/N^{1+c} \leq - ( \area(R_p \cup M_R) - \area (R)) \leq 0$. Therefore,
		\begin{align*}
			-2\alpha/N^{1+c} \leq (\area(L_p \cup M_L) - \area (L)) - ( \area(R_p \cup M_R) - \area (R)) \leq 2\alpha/N^{1+c}
		\end{align*}
	\end{proof}
	
	Recall that, in an \eps-\spizza solution, at most $2n$ lines can intersect $\trngl{ABC}$ and $\trngl{ABC}_p$ (even though the standardized version of the problem requires at most $n$ straight lines, as we showed in \cref{thm:straight-pizza-PPA}, \ppa/-hardness holds even for at most $n+n^{1-\delta}$ lines for any constant $\delta \in (0,1]$). Similarly for an \eps-\scut-\pizza solution, since its \scut-path comprises of at most $n-1$ turns, i.e., $n$ straight line segments (and again, by \cref{thm:hvu-pizza-ppa-h} \ppa/-hardness holds even for at most $n+n^{1-\delta}$ line segments for any constant $\delta \in (0,1]$) By inductively applying \cref{clm: pix-intersection} and \cref{clm: pix-gap} $N$ times (recall that $N \geq 2n$), we get the following.
	\begin{lemma}\label{lem: pix-n-gap}
		Let at most $2n$ straight lines intersect $\trngl{ABC}_p$, and the side of each pixel be $\alpha/15N^{1+c}$ for any $c > 0$. Also, let $M$ be the set of its pixels that are intersected by the lines, and $M_L, M_R$ be an arbitrary partition of $M$. Then, $\area(M) \leq \alpha/5N^c$, and furthermore, $|(\area(L) - \area(R)) - (\area(L_p \cup M_L) - \area(R_p \cup M_R) )| \leq  2\alpha/N^c$.
	\end{lemma}
	
	\myparagraph{Turning pixels into points in general position with unique $x$- and $y$-coordinates.}
	So far, for \eps-\scut-\pizza and \eps-\spizza, with $q \in \{n, 2n\}$ mass distributions, respectively, we have described how to turn each distribution $\mu_i$, $i \in [q]$ into a pixelated version of it. 
	We will now turn each of those pixels into a set of points. Let $\mu_i$ consist of $b$ many polygons, and recall that each polygon has its own weight $w_{i,j}>0$, $j \in [b]$.
	Let $T^{i,j} \in \mathcal{F}_i$ be a non-obtuse triangle belonging to the $j$-th polygon of $\mu_i$, and $\mathcal{F}_i$ be the set of these triangles composing $\mu_i$. Suppose a pixel contains triangles (whose area is strictly positive) \{$T^{i,k} \}_{k \in D}$ for some $D \subseteq [b]$, and let us denote $w^{i}_{\max} := \max_{k \in D} w_{i,k}$. Observe that, due to our assumption that the mass distributions are normalised, we have $1 = \sum_{T^{i,j} \in \mathcal{F}_i} w_{i,j} \cdot \area(T^{i,j}) \geq \sum_{T^{i,j} \in \mathcal{F}_i} w_{i,j} \cdot \alpha$, therefore, 
	\begin{align}\label{eq: up-bound-w}
		w^{i}_{\max} \leq \sum_{T^{i,j} \in \mathcal{F}_i} w_{i,j} \leq 1/\alpha .
	\end{align}
	We will place $\ceil{w^{i}_{\max} \cdot N^c}$ points at the pixel's bottom-left corner, that is, all having the same position. Each of the pixels of mass distribution $i$ has no weight, as desired, and they form a set $P_i$. Notice that each of the $P_i$'s we created contains at most $\frac{2N^c}{t^2 \alpha} \leq \frac{512N^{2+3c}}{\alpha^3}$ points, i.e., polynomially many in the instance's description size and $1/\alpha$. That is because its pixels can be at most $\frac{1}{t} \cdot \frac{1}{t} \leq \left( \frac{15N^{1+c} + 1}{\alpha} \right)^2  \leq \left( \frac{16N^{1+c}}{\alpha} \right)^2 \leq \frac{256N^{2+2c}}{\alpha^2}$ many, with each pixel containing  at most  $\ceil{w^{i}_{\max} \cdot N^c} \leq \frac{N^c}{\alpha} + 1 \leq \frac{2 N^c}{\alpha}$ points.
	
	Observe, however, that in the discrete version of the pizza sharing instance we created, the points of the $q \in \{n, 2n\}$ point sets lie on vertices of a square grid with edge length $t = \frac{1}{\ceil{15N^{1+c}/\alpha}}$. These points are not guaranteed to be in general position or with unique $x$- and $y$-coordinates, and therefore, not all solutions of that instance can be translated back to a solution of the original corresponding continuous version of the instance. For the rest of this section, we will show how to turn this instance into one where the points in $P_1 \cup \dots \cup P_{q}$ are in general position and with unique $x$- and $y$-coordinates. Additionally, we will ensure that each point remains in its original pixel. 

    For each pixel, let us create a $k \times k$ square grid with edge length $\frac{t}{k+1}$, where $k := \frac{48 n^2 N^{2c}}{t^6 \alpha^2}$, which is placed at the center of each pixel. 
    The purpose is to place each point of $P_1 \cup \dots \cup P_{q}$ belonging to a pixel on the pixel's corresponding grid, such that all the aforementioned points are in general position and have unique $x$- and $y$-coordinates. 

    As shown above, for every $i \in [q]$, $|P_i| \leq \frac{2N^c}{t^2 \alpha}$. Therefore, the points that a pixel can contain are at most $q \cdot \frac{2N^c}{t^2 \alpha} \leq \frac{4nN^c}{t^2 \alpha} =: G$ many. Suppose now we place the points of the pixel in distinct vertices of the corresponding grid. 
    \begin{enumerate}
        \item Any pair of those defines a line, and to satisfy the ``general position'' condition, that line should not intersect any other point that is placed on any other pixel's grid. There are at most $\binom{G}{2}$ many such lines. For each line, we will forbid the placement of other points on it by removing the grid vertices it intersects throughout all pixels' grids. This means that, for each line, at most $k$ grid vertices have to be removed due to a pair of points of a single pixel. Therefore, overall, at most $\binom{G}{2} \cdot k \cdot \frac{1}{t^2}$ grid vertices have to be removed to satisfy the ``general position'' condition. 
        \item Each point defines one horizontal and one vertical line that intersects it. To satisfy the ``uniqueness of $x$- and $y$-coordinates'' condition, we have to forbid any other point from being placed on these two lines. To do so, it suffices that among all pixels' grids we remove the vertices that are intersected by these two lines. Each line removes at most $k$ grid vertices in a single pixel, so overall, at most $G \cdot 2 \cdot k \cdot \frac{1}{t^2}$ grid vertices have to be removed to satisfy the ``uniqueness of $x$- and $y$-coordinates'' condition.
    \end{enumerate}
    In total, at most $\binom{G}{2} \cdot k \cdot \frac{1}{t^2} + G \cdot 2 \cdot k \cdot \frac{1}{t^2} = \frac{k}{t^2}\frac{G^2 + 3G}{2} \leq \frac{2kG^2}{t^2}$ grid vertices have to be removed from each pixel's grid in order to respect the above conditions. For each pixel's grid, since it initially contained $k^2$ vertices, its remaining vertices that can be used for points of the pixel to be placed on are at least $k^2 - \frac{2kG^2}{t^2} = \frac{9 G^4}{t^4} - \frac{6 G^4}{t^4} = \frac{3 G^4}{t^4} \geq G$, where the first equality comes by our choice of $k := \frac{48 n^2 N^{2c}}{t^6 \alpha^2} = \frac{3 G^2}{t^2}$. Recall that each pixel contains at most $G$ points, therefore there are enough vertices for them to be placed on.

    So far we have shown that we can slightly perturb each point from the bottom left corner of a pixel so that it remains in the pixel, and furthermore, all points are in general position with unique $x$- and $y$-coordinates. It remains to show that we can do this perturbation in polynomial time. Indeed, it is easy to check that the following procedure achieves this task and requires only polynomially many steps: choose an arbitrary pixel and an arbitrary grid vertex in it to place one of the pixel's points on; from all pixels, remove all other grid vertices which have the same $x$- or $y$-coordinate with any of the pair's points (this can be done in polynomial time by exhaustively checking each of the $k^2/t^2$ grid vertices); next, while there is still a point of that pixel which has not been placed on its grid, place it on one of the non-removed vertices of its grid, then remove every other vertex from all pixels if (i) they have the same $x$- or $y$-coordinate, or (ii) if they are colinear with the pair formed by the new point and any of the older points (this again can be done by checking colinearity in polynomial time for all grid vertices); repeat the previous step for all pixels' points until all points from all pixels (polynomially many) have been placed on the respective grids.

    \begin{figure}[htbp]
			\begin{subfigure}[b]{0.47\textwidth}
				\centering
				\includegraphics[width=0.8\linewidth]{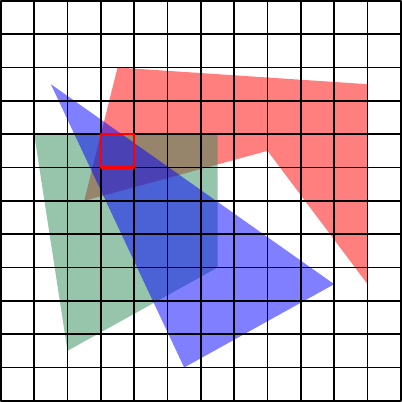}
				\caption{An example of the pixelated masses. All three masses overlap on the red pixel.} \label{fig:pixels-1}
			\end{subfigure}%
			\hspace*{\fill}   
			\begin{subfigure}[b]{0.47\textwidth}
				\centering
				\includegraphics[width=0.8\linewidth]{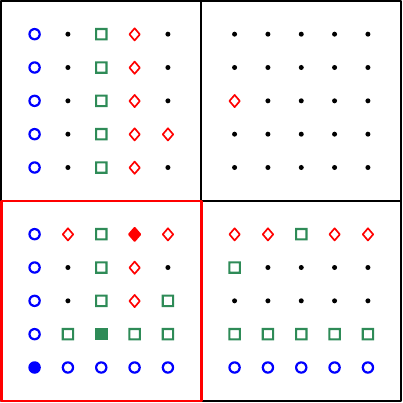}
				\caption{The bottom-left pixel corresponds to the red pixel of \cref{fig:pixels-1}.} \label{fig:pixels-2}
			\end{subfigure}%
			\caption{Turning pixels into points in general position with unique $x$- and $y$-coordinates. In \cref{fig:pixels-2}, three of the red pixel's points are placed one by one in its grid. First, the blue point (shaded disk) is placed on an arbitrary vertex; this ``forbids'' all points that follow to be placed on the blue hollow disks. Then the green point (shaded square) is placed on an arbitrary ``non-forbidden'' vertex; this, in turn, ``forbids'' the green hollow squares. Finally, the red point (shaded rhombus) is placed on one of the remaining vertices; this ``forbids'' the red hollow rhombuses. The black dots denote ``non-forbidden'' vertices where the next available point from the corresponding pixel can be placed on.}
			\label{fig:pixels}
		\end{figure}

    We have proven the following.
 
    \begin{lemma}\label{lem: skewed-grid}
        The perturbed points are in general position, they have unique $x$- and $y$-coordinates, and each has remained in its initial pixel.
    \end{lemma}

    We are now ready to prove the main lemma of this section.
	\begin{proof}[Proof of \cref{lem: cont-to-discr}]
		Consider the input of an \eps-\scut-\pizza or an \eps-\spizza instance, meaning, the description of $q \in \{n, 2n\}$ sets, respectively, of weighted polygons with holes on $[0,1]^2$. By definition, its size is $N \in \Omega(n)$. In time $\poly(N)$, we perform a triangulation of each polygon into non-obtuse triangles, therefore, in total we require $\poly(N)$ time for this task. Then, we perform the ``pixelation'' procedure, which requires, for each of $q$ mass distributions, checking whether each of $1/t^2 = \ceil{225N^{2+2c}/\alpha^2}$ pixels has a non-empty intersection with a triangle. This task can be performed in $\poly(N, 1/\alpha)$ time, since $c>0$ is a fixed constant. Next, the points of each point set $P_i$ created from the respective pixels are perturbed so that they are in general position, and they have unique $x$- and $y$-coordinates (\cref{lem: skewed-grid}). This procedure has as a result that Output (a) of \cref{def:discrete-straight-pizza} and \cref{def:discrete-sc-pizza} cannot be produced. Finally, notice that the number of points to be fairly divided is $\poly(N, 1/\alpha)$.
		
		We claim that the lines $\ell_1, \dots, \ell_m$ for some $m \leq 2n$ that are a solution to \eps-\dspizza or the line segments of the \scut-path solution of \eps-\dpizza (recall that in \cref{thm:straight-pizza-PPA} and \cref{thm:hvu-pizza-ppa-h} we allowed almost $2n$ lines and line segments, respectively) are also a solution to $(\eps - \eps')$-\spizza and $(\eps - \eps')$-\scut-\pizza, respectively, for any $\eps, \eps'$ with $6/N^c \leq \eps' < \eps \leq 1$, where $c>0$ is a constant.
		What remains is to show the correctness of this statement. Notice that, after ``pixelation'', we placed a set of at most $2 N^c/\alpha$ points at the bottom-left corner of the corresponding pixel. Then, we perturbed each point such that it remained in its initial pixel while ensuring that all points are in general position and unique $x$- and $y$-coordinates.
		
		Consider now a line (resp. a line segment) $\ell$ that is part of a solution of the $(\eps - \eps')$-\dspizza (resp. $(\eps - \eps')$-\dpizza) instance, and intersects the corresponding non-obtuse triangle from mass distribution $i \in [q]$, where $q = 2n$ (resp. $q = n$). The triangle has been pixelated, and its corresponding points belonging to $P_i$ have been created. As we discussed above, each point is inside its corresponding pixel. The line that cuts through the triangle can be thought of as intersecting a set $M$ of the pixels of the triangle's pixelated version. The points that correspond to the pixels of $L_p$ are clearly in the $L$-part of the cut; the points that correspond to the pixels of $R_p$ are clearly in the $R$-part of the cut. No matter what part of the cut the points corresponding to $M$ join, \cref{clm: pix-gap} and \cref{lem: pix-n-gap} apply. In the aforementioned results, notice that $L_p \cup M_L$ and $R_p \cup M_R$ correspond to regions that are defined by whole pixels, meaning that they are the union of pixel regions. Therefore, their respective areas are of the form $k_L \cdot t^2$ and $k_R \cdot t^2$, where $k_L , k_R \in \naturals$ represent the number of pixels on each part of the cut. By construction, if our triangle at hand belongs to $\mu_i$ and has weight $w_{i,j}$, then each of the $k_L$ pixels contains at least $\ceil{w_{i,j} \cdot N^c}$ points, and similarly for $k_R$. 
		
		Suppose we have turned the \eps-\spizza (resp. \eps-\scut-\pizza) to $(\eps - \eps')$-\dspizza (resp. $(\eps - \eps')$-\dpizza) as described above, so that all points are in general position (resp. have unique $x$- and $y$-coordinates). A solution for the latter problems always exists due to \cite{S21} and \cref{lem: discr-to-cont-sc-pizza}. Let us have a solution of any of the latter two problems, that is, a set of lines (resp. line segments) $\ell_1, \dots, \ell_m$ for $m \leq 2n$, that partition $[0,1]^2$ to \rplus and \rminus and for every $i \in [q]$, where $q \in \{ 2n, n \}$, we have $\left| |P_i \cap \rplus| - |P_i \cap \rminus| \right| \leq (\eps - \eps') \cdot |P_i|$. 
		
		Recall that $T^{i,j} \in \mathcal{F}_i$ is a non-obtuse triangle which belongs to the $j$-th polygon of $\mu_i$, and $\mathcal{F}_i$ is the set of such triangles that compose $\mu_i$. By $T^{i,j}_p$ we denote the pixelated version of $T^{i,j}$, while $M^{i,j}_{+}, M^{i,j}_{-}$ are $T^{i,j}_p$'s respective parts of the pixels intersected by lines, that join the $\rplus$ and the $\rminus$ sides, respectively.
		From our earlier analysis, we have an upper bound of $\frac{2N^c}{t^2 \alpha}$ for $|P_i|$, however, here we need a better one, namely where a factor of $\alpha$ is removed from the denominator. To that end, we will use the fact that $\sum_{T^{i,j} \in \mathcal{F}_i} w_{i,j} \cdot \area(T^{i,j}) = 1$, or equivalently, $\sum_{T^{i,j} \in \mathcal{F}_i} w_{i,j} \cdot N^{c} \cdot \area(T^{i,j}) = N^{c}$. We have
		\begin{align*}
			\sum_{T^{i,j} \in \mathcal{F}_i} w_{i,j} \cdot N^{c} \cdot \area(T^{i,j}) + \sum_{T^{i,j} \in \mathcal{F}_i} w^{i}_{\max} \cdot N^{c} \cdot \frac{\alpha}{N^{1+c}} \cdot \area(T^{i,j}) \leq \\ \leq N^{c} + w^{i}_{\max} \cdot N^{c} \cdot \frac{\alpha}{N^{1+c}},
		\end{align*}
		and since $\sum_{T^{i,j} \in \mathcal{F}_i} \area(T^{i,j}) \leq 1$, we get
		\begin{align*}
			\sum_{T^{i,j} \in \mathcal{F}_i} \ceil{w_{i,j} \cdot N^{c}} \cdot \area(T^{i,j}) + \sum_{T^{i,j} \in \mathcal{F}_i} \ceil{w^{i}_{\max} \cdot N^{c}} \cdot \frac{\alpha}{N^{1+c}} \cdot \area(T^{i,j}) \leq \\ \leq N^{c} + 1 + \left( w^{i}_{\max} \cdot N^{c} + 1 \right) \cdot \frac{\alpha}{N^{1+c}} .
		\end{align*}

		Now observe that, after pixelation, only the pixels at the boundary of each triangle $T^{i,j}$ can correspond to $\ceil{w^{i}_{\max} \cdot N^{c}}$ points instead of $\ceil{w_{i,j} \cdot N^{c}}$. Therefore, using the notation of \cref{eq: extra-pix}, where $S := \area(T^{i,j})$, only at most a fraction $S'/S$ can correspond to $\ceil{w^{i}_{\max} \cdot N^{c}}$ points. So, 
		\begin{align}\label{eq: points-ub}
			|P_i| &\leq \sum_{T^{i,j} \in \mathcal{F}_i} \ceil{w_{i,j} \cdot N^{c}} \cdot \frac{\area(T^{i,j})}{t^2} + \sum_{T^{i,j} \in \mathcal{F}_i} \ceil{w^{i}_{\max} \cdot N^{c}} \cdot \frac{\alpha}{N^{1+c}} \cdot \frac{\area(T^{i,j})}{t^2} \nonumber  \\
			&\leq \frac{1}{t^2} \cdot \left( N^{c} + 1 + \left( w^{i}_{\max} \cdot N^{c} + 1 \right) \cdot \frac{\alpha}{N^{1+c}} \right)  \nonumber \\
			&\leq \frac{1}{t^2} \cdot \left( N^{c} + 1 + \left( \frac{N^{c}}{\alpha} + 1 \right) \cdot \frac{\alpha}{N^{1+c}} \right)   \nonumber \\
			&\leq \frac{1}{t^2} \cdot \left( N^{c} + 1 + \frac{2}{N} \right)  \nonumber \\
			&\leq \frac{1}{t^2} \cdot \left( N^{c} + 3 \right) ,
		\end{align}
		where the second to last inequality comes from the fact that $1 \leq N^{c}/\alpha$.
		
		Putting everything together, we have
		\begin{align*}
			\left| \mu_i(\rplus) - \mu_i(\rminus) \right| &= \left| 
			\sum_{T^{i,j} \in \mathcal{F}_i} w_{i,j} \cdot \area(T^{i,j} \cap \rplus) - \sum_{T^{i,j} \in \mathcal{F}_i} w_{i,j} \cdot  \area(T^{i,j} \cap \rminus) \right| \\
			&= \left| \sum_{T^{i,j} \in \mathcal{F}_i} w_{i,j} \cdot \left( \area(T^{i,j} \cap \rplus) - \area(T^{i,j} \cap \rminus) \right) \right| \\
			&\leq \left| \sum_{T^{i,j}_p \in \mathcal{F}_i}  w_{i,j} \cdot (\area(T^{i,j}_p \cup M^{i,j}_{+}) - \area(T^{i,j}_p \cup M^{i,j}_{-})) \right| + \frac{2\alpha}{N^{c}} \cdot \sum_{T^{i,j}_p} w_{i,j} \\
			&\leq \left| \sum_{T^{i,j}_p \in \mathcal{F}_i}  w_{i,j} \cdot \left( \frac{|P_{i,j} \cap \rplus|}{\ceil{w_{i,j} \cdot N^{c}}} \cdot t^2 - \frac{|P_{i,j} \cap \rminus|}{\ceil{w_{i,j} \cdot N^{c}}} \cdot t^2 + 1 \cdot t^2 \right) \right| + \frac{2}{N^{c}} \\
			&\leq t^2 \cdot \left| \frac{1}{N^{c}} \cdot \sum_{T^{i,j}_p \in \mathcal{F}_i} |P_{i,j} \cap \rplus| - \frac{1}{N^{c}} \cdot \sum_{T^{i,j}_p \in \mathcal{F}_i} |P_{i,j} \cap \rminus| \right| + \frac{t^2}{\alpha}  + \frac{2}{N^{c}} \\
			&\leq t^2 \cdot \frac{1}{N^{c}} \cdot \left| |P_i \cap \rplus| - |P_i \cap \rminus| \right| + \frac{t^2}{\alpha} + \frac{2}{N^{c}} \\
			&\leq t^2 \cdot \frac{1}{N^{c}} \cdot (\eps - \eps')|P_i| + \frac{t^2}{\alpha}  + \frac{2}{N^{c}}  \\
			&\leq (\eps - \eps') \cdot \frac{1}{N^{c}} \cdot \left( N^{c} + 3 \right) + \frac{\alpha}{225 N^{2+2c}} + \frac{2}{N^{c}}  \qquad \text{(by \cref{eq: points-ub})} \\
			&\leq \eps - \eps' + \frac{3}{N^{c}} + \frac{\alpha}{225 N^{2+2c}} + \frac{2}{N^{c}}  \\
			&\leq \eps - \eps' + \frac{6}{N^{c}} \\
			&\leq \eps ,
		\end{align*}
		where the first inequality is acquired by the reverse triangle inequality in combination with \cref{lem: pix-n-gap}, the second inequality is due to the fact that the number of points in a pixel of $T^{i,j}_p$ is $\ceil{w^{i}_{\max} \cdot N^c}$, where $w^{i}_{\max} \geq w_{i,j}$, by definition, and the last inequality is by definition of $\eps'$.    
	\end{proof}

	\section{Exact computation of polygons' positive measure given a \scut-path} \label{app: exact-computation-measures}
	
	Consider an input of exact \scut-\pizza, i.e., the one of \cref{def:eps-hv-pizza}, where we additionally restrict the mass distributions to be weighted polygons with holes in the representation form described in \cref{sec:preliminaries}. We have so far ensured that our polygons are in $[0,1]^2$. For the computation of function $f$ described in the proof of \cref{thm: sc-pizza_exists} we need a way of computing the ``positive'' measure of a polygon, as dictated by a given $\vec{z} \in S^d$ (and its induced \scut-path). To simplify this computation, we will first use some well-known algorithm (e.g., \cite{GareyJPT78, AsanoAP86, Mehlhorn84, GhoshM91}) that triangulates polygons with holes in $O(N\log N)$ time (which is optimal), where $N$ is the input size, without inserting additional vertices. Then, we will triangulate the polygon further to end up with only non-obtuse triangles, which can be decomposed into axis-aligned right-angled triangles. \cref{alg: polygon_decomp_into_triangles} describes the two preprocessing steps. To prove the algorithm's correctness we need to obtain the preliminary results of the following section.
    
	\subsection{Computing areas of polygons via axis-aligned right-angled triangle decomposition} \label{app: axis-aligned-tr}
	
	Here we first show how an arbitrary triangle can be decomposed into right-angled triangles whose right angle is additionally axis-aligned. 

    Recall that the Borsuk-Ulam function $f : S^d \to \reals^n$ we defined in the proof of \cref{thm: sc-pizza_exists}, namely, $f(\vec{z}) = \mu(\rplus;\vec{z})$ captures the \rplus part of each of $n$ measures when cut by a \scut-path (induced by $\vec{z}$) with $d-1$ turns. From this, it is apparent that we need to be able to compute parts of the area of a polygon, depending on where the \scut-path cuts it. We first need to preprocess the input by: (i) decomposing each polygon into non-obtuse triangles (\cref{prop: non-obtuse_triangle_decomp}), and (ii) decomposing each such triangle into 5 axis-aligned right-angled triangles (\cref{prop: ax-al_ri-tr_decomp}). Then, we provide a Borsuk-Ulam function which, given a \scut-path, captures the measure found on the $\rplus$ region. The aforementioned decomposition, allows our function to be relatively simple, in the sense that it only needs to consider a single shape, that of axis-aligned right-angled triangles.
 
    We start by showing a simple polynomial time routine that achieves the first decomposition step.

    \begin{proposition} \label{prop: non-obtuse_triangle_decomp}
		Any polygon with holes can be decomposed into non-obtuse triangles in polynomial time.
	\end{proposition}

    \begin{proof}
    We first use a standard polynomial-time algorithm to triangulate the given polygon, for example, the technique of \cite{AsanoAP86}. Next, we check the obtuseness of each triangle $\trngl{ABC}$ of the triangulation by computing the squared lengths of its sides $AB^2, BC^2, AC^2$ (each is rational; a sum of squares of rationals), taking the largest one, w.l.o.g. $AC^2$ and then checking whether $AB^2 + BC^2 < AC^2$. If the inequality is not true then $\trngl{ABC}$ is non-obtuse and we proceed with the next available triangle. Otherwise, we add the line segment $BD$ that starts from $B$ and ends at $D$ on side $AC$, where $\widehat{BDC} = \widehat{ADB} = 90^{\circ}$.\footnote{
		We will denote by $\trngl{ABC}$ a triangle with vertices $A,B,C$ and, when clear
		from context, we will also use the same notation to indicate the \textit{area}
		of the triangle. Two intersecting line segments $AB$, $BC$ define two
		\textit{angles}, denoted $\widehat{ABC}$ and $\widehat{CBA}$. The order of the
		vertices implies a direction of the segments, i.e., in the former angle we have
		$AB$, $BC$ and in the latter we have $CB$, $BA$. We consider the direction of
		the segments and define the angle to be the intersection of the \textit{left}
		half-spaces of the segments. Therefore $\widehat{ABC} = 360^\circ -
		\widehat{CBA}$. This order will not matter if clear from context (e.g., in triangles).} The coordinates $(x_D, y_D)$ of $D$ are rationals since they are the solution of the following two equations: (a) one that dictates that $D$ is on $AC$: $\frac{y_{D} - y_{A}}{x_D - x_A} = \frac{y_A - y_C}{x_A - x_C}$, and (b) one that captures the fact that $BD$ and $AC$ are perpendicular: $\frac{y_B - y_D}{x_B - x_D} \cdot \frac{y_A - y_C}{x_A - x_C} = -1$. In fact, 
	\begin{align*}
		&x_D = \frac{Num}{Den}, \quad \text{and} \quad y_D = \frac{y_A - y_C}{x_A - x_C}(x_D - x_ A) + y_A ,
	\end{align*}
	\begin{align*}
		\text{where} \quad Num =& (y_A - y_D) [(y_B -y_A)(x_A - x_C) + (y_A - y_C) x_A] + x_B (x_A - x_C)^2 , \\
		\text{and} \quad Den =& (x_A -x_C) (x_A - x_C + y_A - y_C)	
	\end{align*}
    This results in two right-angled triangles $\trngl{ABD}$ and $\trngl{BCD}$. We then proceed with the next available triangle of the triangulation, until there is none left. 

    It is easy to see that the triangulation results in polynomial many triangles, and the above check and potential split of each triangle requires at most polynomial time. 
    \end{proof}
	
	At this point, the triangulation of each polygon consists of non-obtuse triangles (see \cref{fig:combined-a}).
	The next proposition decomposes further each non-obtuse triangle into axis-aligned right-angled triangles.
	
	\begin{proposition} \label{prop: ax-al_ri-tr_decomp}
		The area of any non-obtuse triangle can be computed using the areas of five axis-aligned right-angled triangles.
	\end{proposition}
	
	\begin{proof}
		Consider a non-obtuse triangle $\trngl{ABC}$. A proof by picture is presented in \cref{fig:combined-b}, where we draw a segment from the top-left corner to the bottom-right one, and $\area(\trngl{ABC}) = \area(\trngl{XYB}) + \area(\trngl{XBZ}) - \area(\trngl{AYB}) - \area(\trngl{XAC}) - \area(\trngl{CBZ})$.
		
		The proof is immediate if we show that every non-obtuse triangle $\trngl{ABC}$ can be inscribed inside a rectangle such that all of its vertices touch the rectangle's perimeter and one of them touches a corner of the rectangle while each of the other two vertices touches one of the rectangle's non-incident sides to that corner. 
        To see this, consider a rectangle of minimum perimeter, which contains $\trngl{ABC}$ and its sides are parallel to the axes. Since its perimeter is minimum, each side touches at least one of the triangle's vertices, otherwise the perimeter could be reduced. And since the triangle has only three vertices, at least one of them has to be touching two sides of the rectangle, i.e., a corner of the rectangle. If only one triangle vertex touches a corner of the rectangle, then each of the other two vertices touches one of the non-incident sides of the rectangle's corner, otherwise the rectangle's perimeter can be reduced. If two triangle vertices are on corners of the rectangle, then the third vertex has to be on another corner (i.e., it is a right-angled triangle); otherwise, either the rectangle does not have minimum perimeter (the two corners have a common side), or it is obtuse (the two corners do not have a common side), both being contradictions. 
        
        Due to the above, the coordinates to each of the rectangle's corners are $(x_{L},y_{L}), (x_{L},y_{H}), (x_{H},y_{H})$, and $(x_{H},y_{L})$, where $x_{L}, x_{H}, y_{L}, y_{H}$ denote the minimum and maximum $x$- and $y$-coordinates of $\trngl{ABC}$'s vertices, respectively.
	\end{proof}

    Finally, using the above auxiliary results, \cref{alg: polygon_decomp_into_triangles} shows how to decompose each given polygon into axis-aligned right-angled triangles in polynomial time.
    
	\begin{algorithm}[h!]
		\caption{Preprocessing: decomposing polygons into axis-aligned right-angled triangles}  
		\label{alg: polygon_decomp_into_triangles}  
		\begin{algorithmic}[1]
			\REQUIRE{A polygon $P$ represented by its ordered vertices (see \cref{sec:preliminaries}).}
			\ENSURE{A set $H$ consisting of 5-tuples; each 5-tuple $s \in \left[ |H| \right]$ corresponds to a non-obtuse triangle $T_s$ such that $\bigsqcup_{s \in \left[ |H| \right]} \area(T_s) = \area(P)$; each tuple $( T_s^1, T_s^2, T_s^3, T_s^4, T_s^5 )$ consists of $5$ axis-aligned right-angled triangles such that $\area(T_s) = \area(T_s^1) + \area(T_s^2) - \area(T_s^3) - \area(T_s^4) - \area(T_s^5)$}
			
    		  \medskip
    		  \textbf{Preprocessing step 1:}
                \STATE{Run a poly-time algorithm to triangulate the given polygon (e.g., \cite{AsanoAP86}), and let $G$ be the set of the resulting triangles.}
    		\WHILE{there is an unchecked triangle $T \in G$}
                    \STATE{Check the obtuseness of $T$}
    			\IF{$T$ is obtuse}
                        \STATE{Split it into two right-angled triangles $T_{I}, T_{II}$ (Proof of \cref{prop: non-obtuse_triangle_decomp}).}
                        \STATE{$G \gets G \cup \{ T_{I}, T_{II} \}$}
    			\ENDIF
                \ENDWHILE
                
                \medskip
                \textbf{Preprocessing step 2:}
                \STATE{$H \gets \emptyset$}
                \WHILE{$G \neq \emptyset$}
                    \STATE{Consider a (non-obtuse) triangle $T \in G$.}
                    \STATE{Define the tuple $r := \left( (x_{L},y_{L}), (x_{L},y_{H}), (x_{H},y_{H}), (x_{H},y_{L}) \right)$, where $x_{L}, x_{H}, y_{L}, y_{H}$ denote the minimum and maximum $x$- and $y$-coordinates of $T$'s vertices, respectively.}
                    \STATE{Find a point in $r$ which corresponds to a non-right angle of $T$ and name that vertex $B$ (Guaranteed by the proof of \cref{prop: ax-al_ri-tr_decomp}).}
                    \STATE{Let $B$ be element $r(i)$ for some $i \in [4]$, and name the following points: $Y = r(i \pmod{4} + 1)$, $X = r(i+1 \pmod{4} + 1)$, and $Z = r(i+2 \pmod{4} + 1)$.}
                    \STATE{Name $A$ the vertex of $T$ located on the segment $YX$, and $C$ the vertex of $T$ on the segment $XZ$ (Both guaranteed to be in these segments by the proof of \cref{prop: ax-al_ri-tr_decomp}).}
                    \STATE{Define the tuple $v \gets (\trngl{XYB}, \trngl{XBZ}, \trngl{AYB}, \trngl{XAC}, \trngl{CBZ})$.}
                    \STATE{$H \gets H \cup \{ v \}$}
                    \STATE{$G \gets G \setminus \{ T \}$}
                \ENDWHILE			
		\end{algorithmic}
	\end{algorithm}

	\subsection{Constructing the Borsuk-Ulam function} \label{app: BU-func_comp}
	
	Here we show how to construct the Borsuk-Ulam function $f: S^d \to \reals^n$ given $n$ sets of weighted polygons, where $d, n \in \naturals$. We will focus on an arbitrary colour $i \in [n]$ and present the coordinate $f_{i}$. As discussed earlier, our function will capture the measure $\mu_i$ in the \rplus region of any given \scut-path. Furthermore, after the preprocessing step achieved by \cref{alg: polygon_decomp_into_triangles}, the function suffices to be able to capture the measure of simple shapes, namely axis-aligned right-angled triangles.
	
	Consider the $\tau \geq 1$ weighted polygons of the $i$-th colour, and let us 
	focus on a particular polygon $t \in [\tau]$. We have triangulated the polygon 
	into $m_t$ non-obtuse triangles. Consider one such triangle $T_s$, corresponding to some $s \in [m_t]$, and the \textit{virtual triangles} $T_{s}^{1}, T_{s}^{2}, T_{s}^{3}, T_{s}^{4}, 
	T_{s}^{5}$, which are the five axis-aligned right-angled triangles 
	that \cref{alg: polygon_decomp_into_triangles} gave as output (also see \cref{fig:combined-b}). W.l.o.g. we consider $T_{s}^1$ 
	and $T_{s}^2$ to be the \textit{positively contributing triangles} and the rest 
	to be the \textit{negatively contributing triangles}. For each of them, we will 
	be computing the positive measure determined by the \scut-path induced from the given point $\vec{z} \in S^d$ (see proof of \cref{thm: sc-pizza_exists}). By the axis-aligned right-angled triangle 
	decomposition described in the proof of \cref{prop: ax-al_ri-tr_decomp}, it suffices to show how to compute parts of areas of such a triangle, for all of its four possible \textit{orientations}: $Q_{I}, Q_{II}, Q_{III}, Q_{IV}$, where $Q_o$ is the orientation when, by shifting the triangle so that the vertex of the right angle is on $(0,0)$, the whole triangle is in the $o$-th quadrant.
	
	First, we identify the orientation of our triangle. For a fixed colour $i \in [n]$, for each possible category $Q_o$, $o \in \{ I, II, III, IV \}$ we show how to compute the term that an axis-aligned right-angled triangle $T_{s}^{v}$, $s \in [m_t]$, $v \in [5]$, contributes to the Borsuk-Ulam function $f_{i}(\vec{z})$, where $T_{s}$ is a non-obtuse triangle $\trngl{ABC}$ as in \cref{fig:combined-b}. We will show this for a $Q_I$ triangle with the help of \cref{fig: val-function}. The constructions for $Q_{II}, Q_{III}, Q_{IV}$ are omitted since they are symmetric to it.

    In what follows, for any point $W$ of $[0,1]^2$ we will denote by $x_{W}, y_{W}$ its coordinates. Suppose we are given the $Q_I$ triangle $\trngl{AYB}$ (as in \cref{fig:combined-b}). We are also provided with some \scut-path induced by $\vec{z} \in S^d$, and we focus on the strip $[y_j, y_{j+2}]$ for some $j \in \{ 0, 1, 3, 5, \dots, d-1 \}$ (resp. $j \in \{ 0, 1, 3, 5, \dots, d \}$) when $d$ is even (resp. odd), and $z_{j}, z_{j+1}$ which induce the vertical cut $x_j$ and define an \rplus and an \rminus region of the slice (see the proof of \cref{thm: sc-pizza_exists} for details). We also add all the artificial cuts needed in the bottom and top strips (see \cref{alg: cuts_to_paths}). We are only interested in the part of $\trngl{AYB}$ in the \rplus region, but since this could be either to the left or to the right of $x_j$, let us denote by $\area_j^{\ell}( \trngl{AYB} )$ and $\area_j^{r}( \trngl{AYB} )$ the areas of $\trngl{AYB}$ to the right and left part, respectively, of $x_j$ in the $j$-th slice. 

    \begin{figure}[t]
		\centering
		\includegraphics[scale=0.8]{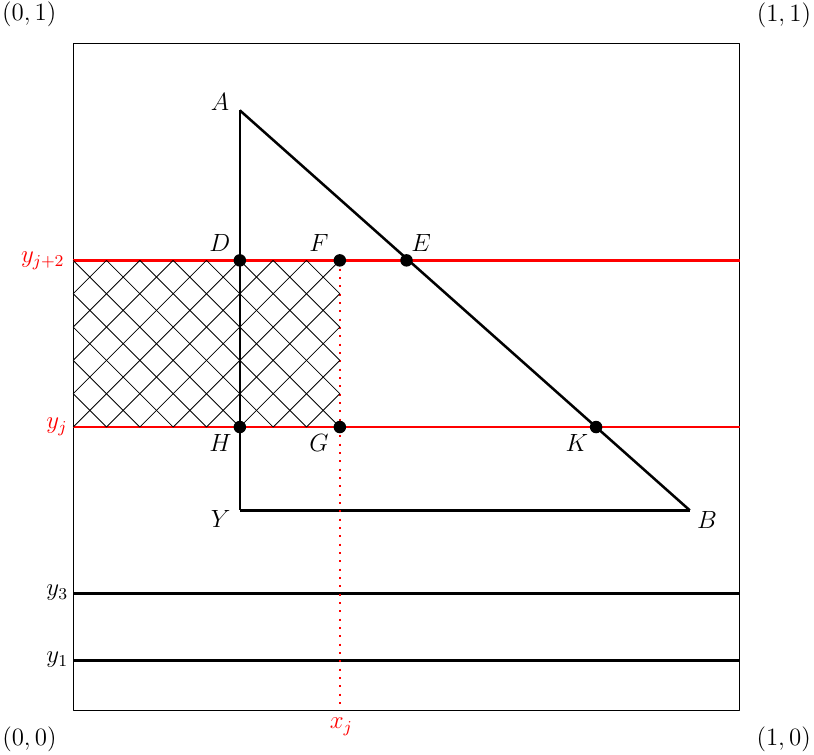}
		\caption{An example of a $Q_{I}$ triangle $\trngl{AYB}$. The given $\vec{z} \in S^d$ has $z_j < 0$ and $z_{j+1} > 0$, therefore, the \rplus part of $\trngl{AYB}$ is $\area_j^{r}( \trngl{AYB} ) = \area(EFGK)$ and its \rminus part is $\area_j^{\ell}( \trngl{AYB} ) = \area(FDHG)$.}
		\label{fig: val-function}
    \end{figure}

    Using \cref{fig: val-function}, we have $\area_j^{\ell}( \trngl{AYB} ) = \area(FDHG)$, and $\area_j^{r}( \trngl{AYB} ) = \area(EFGK)$. We have the following cases. (i) $x_j \in (0,1)$: if $z_j \geq 0$ and $z_{j+1} \leq 0$ (resp. $z_j \leq 0$ and $z_{j+1} \geq 0$), then $\area_j^{\ell}( \trngl{AYB} ) $ belongs to \rplus (resp. \rminus) and $\area_j^{r}( \trngl{AYB} )$ belongs to \rminus (resp. \rplus). (ii) $x_{j} \in \{ 0, 1 \}$: if $z_j + z_{j+1} \geq 0$ (resp. $z_j + z_{j+1} \leq 0$), then the part of $\trngl{AYB}$ that belongs to the $j$-th strip, denoted $\area_j( \trngl{AYB} ) $, belongs entirely to \rplus (resp. \rminus).\footnote{Notice that cases (i) and (ii) include the subcase $z_{j} = z_{j+1} = 0$. However then, the sign(s) of the (possibly two) parts of the $j$-th strip do not matter since the strip has width 0 and therefore does not contribute to the Borsuk-Ulam function.} So, w.l.o.g., we can say that $\area_j^{\ell}( \trngl{AYB} ) $ and $\area_j^{r}( \trngl{AYB} )$ have opposite signs (and it is possible that one of these areas is 0).

    Since the slices $y_j, y_{j+2}$, and cut $x_j$, in general, can have values that do or do not intersect $\trngl{AYB}$, we create truncated versions of them as follows:
    \begin{align*}
        y_{j}^{tr} := \max \{ y_B , \min \{ y_A, y_j \} \} 
        = \begin{cases}
            y_A, & y_j > y_A \\
            y_j, & y_j \in [y_B, y_A] \\
            y_B, & y_j < y_B,
        \end{cases}
    \end{align*}
    and similarly, $y_{j+2}^{tr} := \max \{ y_B , \min \{ y_A, y_{j+2} \} \}$, and $x_{j}^{tr} := \max \{ x_A , \min \{ x_K, x_j \} \}$. Given these, we need to define properly the $y$-coordinate of points $F, G$, and the $x$-coordinate of points $K, E$, so that the points stay on the boundary of the trapezoid $EDHK$. Observe that the line passing from points $A,B$ is described by $y = y_B + \frac{x_B - x}{x_B - x_A}(y_A - y_B)$, or equivalently, $x = x_B - \frac{y - y_B}{y_A - y_B}(x_B - x_A)$, so using these we get:
    \begin{align*}
        y_F &:= \max \left\{ y_B , \min \left\{ y_{j+2}^{tr}, y_B + \frac{x_B - x_{j}^{tr}}{x_B - x_A}(y_A - y_B) \right\} \right\}, \\ 
        y_G &:= \max \{ y_B, \min \{ y_{j}^{tr}, y_F \} \}, \\
        x_K &:= x_B - \frac{y_{j}^{tr} - y_B}{y_A - y_B}(x_B - x_A), \\
        x_E &:= x_B - \frac{y_{j+2}^{tr} - y_B}{y_A - y_B}(x_B - x_A).
    \end{align*}

    Now we are ready to compute the length of our line segments. We have, $EF = \max \{ 0, x_E - x_{j}^{tr} \}$, $GK = \max \{ 0 , x_K - x_{j}^{tr} \}$, $FG = y_F - y_G$, $HK = x_K - x_A$, $DE = x_E - x_A$, and $DH = y_{j+2}^{tr} - y_{j}^{tr}$. Using these, we can compute the quantities of interest:
    \begin{align*}
        \area_j^{r}( \trngl{AYB} ) &= \area(EFGK) = \frac{(GK + EF) \cdot FG}{2}, \\
        \area_j^{\ell}( \trngl{AYB} ) &= \area(FDHG) = \area(EDHK) - \area(EFGK) \\
        &= \frac{(HK + DE) \cdot DH}{2} - \frac{(GK + EF) \cdot FG}{2}.
    \end{align*}

    Using the above, we pick the element from $\left\{ \area_j^{\ell}( \trngl{AYB} ), \area_j^{r}( \trngl{AYB} ) \right\}$ that belongs to \rplus, and let us denote this quantity $p_{s}^{v}(j)$. This quantity represents the part that \emph{only the $j$-th strip} contributes to the positive measure of the Borsuk-Ulam function due to triangle $T_s^v$, for some $s \in [m_t]$ and $v \in [5]$. Consequently, the positive measure that the \emph{entire (unweighted) non-obtuse triangle $T_s$} contributes to the Borsuk-Ulam function according to the \scut-path induced by $\vec{z}$ is 
	\begin{align*}
		q_{s} := \sum_{j \in J} \left( p_{s}^{1}(j) + p_{s}^{2}(j) - p_{s}^{3}(j) - p_{s}^{4}(j) - p_{s}^{5}(j) \right),
	\end{align*}
    where $J := \{ 0, 1, 3, 5, \dots, d-1 \}$ (resp. $J := \{ 0, 1, 3, 5, \dots, d \}$) when $d$ is even (resp. odd). 
    
	Finally, recall that colour $i \in [n]$ has $\tau$ many weighted polygons, each of weight $w_{t}$, $t \in [\tau]$. Also, each polygon has been decomposed into $m_t$ many non-obtuse triangles. Then, $i$'s positive measure (i.e., the $i$-th coordinate of the Borsuk-Ulam function) is
	\begin{align*}
		f_{i}(\vec{z}) = \sum_{t=1}^{\tau} w_t \sum_{s=1}^{m_t} q_s.
	\end{align*}

\bigskip	
\section*{Acknowledgements} 
We thank the anonymous reviewers for comments and suggestions that helped simplify some proofs and improve the presentation of the paper. The first author was supported by EPSRC Grant EP/X039862/1 ``NAfANE: New Approaches for Approximate Nash Equilibria''. The second author was supported by EPSRC Grant EP/W014750/1 ``New Techniques for Resolving Boundary Problems in Total Search''.

	\bibliographystyle{alphaurl}
	\bibliography{references}

\end{document}